\pdfoutput=1
\newif\ifAnon
\Anonfalse
\documentclass[letterpaper]{cccg25} 
\usepackage{amsmath, amssymb} 

\usepackage[T1]{fontenc}
\usepackage{tikz}
\usepackage{tikz-3dplot}
\usepackage{xcolor}
\usepackage{cite}
\usetikzlibrary{trees,positioning,decorations.pathreplacing,calc,patterns,patterns.meta,fit,perspective,angles,quotes}
\usepackage{color}
\usepackage{graphicx}
\usepackage{calc}
\usepackage{amsfonts}
\usepackage{xstring}
\usepackage{hyperref}
\usepackage{cleveref}
\usepackage{caption}
\usepackage{subcaption}
\usepackage{ifthen}
\usepackage[misc]{ifsym}
\usepackage{pifont}
\usepackage{mathtools}
\usepackage{trimclip}
\usepackage{multirow}
\usepackage{algorithm}
\usepackage{algorithmicx}
\usepackage{algpseudocodex}

\newcommand{\R}{\mathbb{R}}

\renewcommand{\subsection}[1]{\textbf{#1}.}
\renewcommand{\subsubsection}[1]{\textbf{#1}.}

\definecolor{okabe1}{HTML}{000000}
\definecolor{okabe2}{HTML}{E69F00}
\definecolor{okabe3}{HTML}{56B4E9}
\definecolor{okabe4}{HTML}{009E73}
\definecolor{okabe5}{HTML}{F0E442}
\definecolor{okabe6}{HTML}{0072B2}
\definecolor{okabe7}{HTML}{D55E00}
\definecolor{okabe8}{HTML}{CC79A7}

\newcommand{\xmark}{\textcolor{okabe7}{\ding{55}}}%

\renewcommand{\emph}[1]{\textit{\textbf{#1}}}
\let\epsilon\varepsilon

\makeatletter
\newcommand{\drawCircleAlgorithm}[3][0.3]{%
  \@ifnextchar\bgroup{\@drawCircleAlgorithm{#1}{#2}{#3}}{\@drawCircleAlgorithm{#1}{#2}{#3}{}}%
}
\newcommand{\@drawCircleAlgorithm}[4]{%
  \def\circles{#2}
  \draw[line width=#1,dashed] (0,0) circle (1);
  \foreach \x/\y/\r [count=\i] in \circles {
    \pgfmathtruncatemacro\colorindex{mod(\i-1,7)+2}
    \draw[line width=#1,okabe\colorindex] (\x,\y) circle (\r);
    \pgfmathsetmacro\scalefactor{max(0.4, min(1.2, 2.5*\r))}
    \ifnum\i<#3
      \node[okabe\colorindex,scale=\scalefactor] at (\x,\y) {\i};
    \else\ifnum\i=#3
      \node[okabe\colorindex,scale=\scalefactor] at (\x,\y) {#4};
    \fi\fi
  }
}
\makeatother

\title{The Rectilinear Marco Polo Problem\thanks{This research was supported in part by NSF grant 2212129.}}
\ifAnon
\author{Anonymous Author(s)}
\else
\author{Ofek Gila\thanks{University of California, Irvine, \texttt{ogila@uci.edu}}
    \and
    Michael T. Goodrich\thanks{University of California, Irvine, \texttt{goodrich@uci.edu}}
    \and
    Zahra Hadizadeh\thanks{University of Rochester, \texttt{zhadizadeh99@gmail.com}}
    \and
    Daniel S. Hirschberg\thanks{University of California, Irvine, \texttt{dhirschb@uci.edu}}
    \and
    Shayan Taherijam\thanks{University of California, Irvine, \texttt{staherij@uci.edu}}}

\index{Gila, Ofek}
\index{Goodrich, Michael T.}
\index{Hadizadeh, Zahra}
\index{Hirschberg, Daniel S.}
\index{Taherijam, Shayan}
\fi

\begin{document}

%
%
 \pagestyle{plain}

\maketitle

\begin{abstract}
We study the \emph{rectilinear Marco Polo problem}, which generalizes the
Euclidean version of the Marco Polo problem
for performing geometric localization 
to rectilinear search environments, such as
in geometries motivated from urban settings, and to higher dimensions.
In the
rectilinear Marco Polo problem, there is at least one 
\emph{point of interest} (POI) within distance $n$,
in either the $L_1$ or $L_\infty$ metric, from the origin.
Motivated from a search-and-rescue application,
our goal is to move a \emph{search point}, $\Delta$,
from the origin to a location within distance $1$ of a POI.
We periodically issue \emph{probes} from $\Delta$ out a
given distance (in either the $L_1$ or $L_\infty$ metric)
and if a POI is within the specified distance of 
$\Delta$, then we learn this (but no
other location information).
Optimization goals are to minimize the number of probes and the distance
traveled by $\Delta$.
We describe a number of efficient search strategies
for rectilinear Marco Polo problems
and we analyze each one in terms of the size, $n$, of the search domain, as
defined by the maximum distance to a POI.
\end{abstract}

\section{Introduction}

Gila, Goodrich, Hadizadeh, Hirschberg, and 
Taherijam~\cite{gila2025combinatorialdronesearching} 
introduce the Marco Polo problem, which they motivate
in terms of one or more points of interest (POIs),
thought of as hikers lost in a forest, that
we would like to localize using a mobile search point, $\Delta$.
Each lost hiker is assumed to have a wireless 
device that can respond to probes sent from $\Delta$ to a specified
distance such that if a lost hiker is within that distance of $\Delta$,
then the search algorithm will receive a positive response.
Such probes use up power, of course, both for $\Delta$ and for a POI's
tracking device; hence, the goal is to devise a search algorithm
for $\Delta$ and a sequence of probes that minimizes the number of probes
needed to locate a POI to within a distance of~1.

In the formulation of Gila {\it et al.}~\cite{gila2025combinatorialdronesearching}, the underlying
geometry for the Marco Polo problem is Euclidean, such that $\Delta$ can move
unrestrictedly in any direction and probes are circles, which seems
reasonable for searching in a forest but not for searching in
an urban environment where distance is more accurately abstracted
as being rectilinear.
In this paper, we are interested in studying a rectilinear
version of the
Marco Polo problem.

As a colorful motivation of the two-dimensional version of the
rectilinear Marco Polo problem, 
suppose one or more people
have been kidnapped and are being held in one or more secret 
points of interest (POIs)
in a city 
(like New York, Chicago, or Toronto) whose streets are essentially grids.
A mobile search point, $\Delta$, can move
to search for them that is restricted to flying or driving along
rectilinear paths (since it cannot fly or drive through buildings). 
Each kidnap victim at a POI has a hidden electronic device that can respond to probes
from $\Delta$, which can issue probes to specified rectilinear 
distances such that if there is a kidnap victim within this distance,
then our search algorithm will learn this.
But the search algorithm does not learn
the direction or distance to the kidnap victim.
The optimization problem
is to minimize the number of probes and/or victim responses, as well
as possibly minimizing the travel distance for $\Delta$.
We are therefore interested
in efficient searching strategies
for rectilinear Marco Polo problems
with analyses in terms of the size, $n$, of the search domain.

We can therefore formulate the rectilinear Marco Polo problem as a computational
geometry problem, where we have at least one \emph{point of interest} (POI) 
at distance at most $n$ from the origin, and we want to move a \emph{search point},
$\Delta$,
to within distance $1$ of a POI, guided by probes.
A probe is a query specified by $\Delta$'s position and a distance, $d$, such
that we learn whether or not a POI is within distance $d$ from 
$\Delta$, in either the $L_1$ or $L_\infty$ metric.
Our optimization goals are to minimize the number of probes and distance
for $\Delta$ to travel to find a POI.

\paragraph{Related Prior Work.}
We are not familiar with any prior work
on the rectilinear Marco Polo problem.
As mentioned above,
Gila, Goodrich, Hadizadeh, Hirschberg, and Taherijam~\cite{gila2025combinatorialdronesearching} 
introduce the Euclidean version
of the Marco Polo problem, where 
search paths are not restricted and 
travel distances and probe distances are measured with the 
Euclidean $L_2$ metric.
For example, 
they provide a number of carefully choreographed travel patterns
and probe strategies, including one that finds a POI with
$3.34\lceil\log n\rceil$ probes and flight distance $6.02n$
and a strategy that uses
$2.53\lceil\log n\rceil$ probes and flight distance $45.4n$.

The Marco Polo problem
is related to combinatorial group testing,
see, e.g., \cite{du1999combinatorial,eppstein2007improved,goodrich2008improved,%
dorfman,covid}.
In this problem,
one is given a set of $n$ items, at most
$d$ of which are ``defective.''
Subsets of the items can be 
pooled and tested as a group, such that if one of the items 
in the pool is defective, then the test for the pool will be positive.
Tests can be organized either adaptively or non-adaptively to
efficiently identify the defective items based on the outcomes
of the tests.
The Marco Polo problem
differs from combinatorial group
testing, however, in that the search space for the Marco Polo problem
is a geometric region and tests must be connected
geometric shapes (like squares), whereas the search space in
combinatorial group testing is for a discrete set of $n$ items
and tests can be arbitrary subsets.

Another related problem in computational geometry 
is the freeze tag 
problem~\cite{arkin2006freeze,hammar2006online,arkin2003improved,bonichon2024euclidean}, 
which has also been studied in the rectilinear 
setting~\cite{bonichon2024freeze,pedrosa2023freeze}, 
where one is interested in moving robot points
in the plane to ``wake up'' a collection of robots.
Also, another related rectilinear combinatorial optimization problem is the 
optimization problem abstracted from the Minesweeper game; 
see, e.g.,~\cite{kaye2000minesweeper}.
There is also considerable prior work
on search-and-rescue algorithms focused on
non-combinatorial solutions, including the use of
continuous monitoring, sophisticated cameras, and 
non-adaptive travel patterns;
see, e.g.,~\cite{albanese,BETTI,MISHRA20201,schedl,tian20,sun18,cenwits,atif21}.

\paragraph{Problem definition.}
In the \emph{rectilinear Marco Polo} problem,
there are $k\ge 1$ entities, which we'll call \emph{points of interest} (POIs),
with unknown positions, at least one of which are within
a distance, $n$, in the $L_1$ or $L_\infty$ distance metric,
of a point, $O$, called the \emph{origin}. 
That is, the search region is a diamond or square in $\R^2$ or a 
octahedron or cube in $\R^3$.

A \emph{probe}, $p(x,y,d)$, is a query that asks if there is any POI
within distance $d$ of the current position, $(x,y)$, 
of a search point, $\Delta$,
measured under a distance metric,
since such a
point is the position at which, e.g., a search algorithm would issue
a probe request to a wireless device of a lost kidnap victim.
The goal is to design a search strategy for $\Delta$ to localize one or more
POIs to within a distance of $1$.
In this paper, we primarily consider the case where $\Delta$ only searches for a
single POI, which can be combined with an \emph{incremental} search strategy
which finds POIs one at a time, for example, using the generalized algorithm
of~\cite{gila2025combinatorialdronesearching}.
We also make no assumptions about the number of POIs, $k$, and their locations,
besides the fact that at least one POI is within distance $n$ from the origin,
referred to the \emph{unbounded} version of the problem.
Finally and most importantly, unlike the paper by Gila {\it et
al.}~\cite{gila2025combinatorialdronesearching}, which focuses solely on the
$L_2$ distance metric, we consider the rectilinear metrics, $L_1$ and
$L_\infty$.



For any search strategy, 
there are a number of ways we can 
measure the effectiveness of the strategy, including:
\begin{itemize}
\item
$P(n)$: the number of probes issued by $\Delta$.
\item
$R_\text{max}$: the maximum number of times a POI
must respond to a probe.
\item
$D(n)$: the total distance traveled by $\Delta$, in a chosen
rectilinear metric, such as $L_1$ or $L_\infty$.
\end{itemize}

\paragraph{Our Results.}
In this paper, we provide a number of efficient algorithms for 
the solving Marco Polo problems,
with algorithms that achieve optimal or near-optimal performance across all
measures.
We begin with a warm-up algorithm which sequentially checks each quadrant of the
search area in the 2D case, and each octant in the 3D case.
While this simple algorithm indeed has a poor worst-case probe performance, we
show that it effectively minimizes the number of responses required by each
POI to just $\lceil \log n \rceil$, regardless of the number of
dimensions.
We then present a pair of more sophisticated algorithms that use a domino-like
recursive search pattern to achieve great probe complexities, to just 
$2 \lceil \log n \rceil + 1$ in 2D and $3 \lceil \log n \rceil + 3$ in 3D, near-optimal
with respect to the lower bounds of 
$2 \lceil \log n \rceil$ and $3 \lceil \log n \rceil$, for 2D and 3D, respectively.

We then focus on minimizing the distance traveled by $\Delta$ with respect to
the distance to the nearest POI, $\delta_\text{min}$, presenting an
algorithm which performs a binary search for each dimension, which we call
\emph{central binary search} (CBS).
A search point, $\Delta$, following this algorithm will travel a distance of at most 
$2 \delta_\text{min} + O(1)$  in 2D, $3 \delta_\text{min} + O(1)$ in 3D, while
still maintaining a near-optimal probe complexity.
Afterwards, we show how to extend our algorithms to higher dimensions, achieving
the orthant algorithm that yields good POI response performance for all
dimensions, and achieving the generalized CBS algorithm, which provides
near-optimal probe complexity and instance-optimal distance performance.
Finally, we present a method to make the probe and response performance of any
algorithm instance-optimal with respect to $\delta_\text{min}$.
In \Cref{sec:experiments}, we include experiments that support our results.

For simplicity, we primarily focus on instances of
the rectilinear Marco Polo problem in the $L_\infty$ metric in this paper,
representing the search area as a square or cube rather than as a diamond or
octahedron in our figures. Likewise, we refer to the search area as a hypercube in
general case rather than as a cross-polytope.
Nevertheless, our results apply equally well to the $L_1$ metric except where
explicitly specified.





\section{Rectilinear Searching Strategies}\label{sec:algs}

We first 
introduce 
a number of algorithms
for the 2D and 3D rectilinear searching problem%
, including more natural algorithms and some
more involved ones, and show how each are
optimal with respect to a different metric.
In general, we assume that there may be multiple POIs, yet we are only
interested in finding one of them, and that $\Delta$ starts its search from the
origin.
We do note, however, that each of our algorithms can be used as a subroutine in
the general method of finding all targets as described
in~\cite{gila2025combinatorialdronesearching}.
We generalize several algorithms to higher dimensions in a later section.

\subsection{Quadrant and Octant Algorithms}
Perhaps the first algorithms that come to mind are those that divide the
original search area into quadrants and octants.
We probe each quadrant from its center, with a probe of half the radius of the
parent region.
Since a POI is known to exist within the search area,
at most three quadrants need to be probed.
And each subsequent \emph{layer} has half the radius of the previous.
As such, there are at most $\lceil \log{n} \rceil$ layers total, and since each
layer may take up to three probes, the total number of probes, $P(n)$, is at most
$3\lceil \log{n} \rceil$ probes.
In the 3D case, there are eight octants in total, where at most seven must
be probed, resulting in $P(n) < 7\lceil \log{n} \rceil$ probes.
See \Cref{fig:quad-search}.

\begin{figure}[ht!]
	\centering
	\resizebox*{0.7\linewidth}{!}{\trimbox{0 3 0 0}{\begin{tikzpicture}
	\fill[pattern={Lines[angle=-45,distance=6]}, pattern color=okabe7, opacity=0.7] (0,1) rectangle (2,2);
	\fill[pattern={Lines[angle=-45,distance=6]}, pattern color=okabe7, opacity=0.7] (1,0) rectangle (2,1);

	\fill[pattern={Lines[angle=-45,distance=6]}, pattern color=okabe7, opacity=0.7] (0,0.5) rectangle (0.5,1);

	\fill[okabe4, opacity=0.4] (0.5,0.5) rectangle (1,1);

	\draw[line width=0.5] (0,1) -- (2,1);
	\draw[line width=0.5] (1,0) -- (1,2);

	\draw[line width=0.25] (0,0.5) -- (1,0.5);
	\draw[line width=0.25] (0.5,0) -- (0.5,1);
	
	\node at (0.5,1.5) {\large 1};
	\node at (1.5,1.5) {\large 2};
	\node at (1.5,0.5) {\large 3};

	\node at (0.25,0.75) {\small 4};
	\node at (0.75,0.75) {\small 5};

	\draw[line width=1] (0,0) -- (2,0) -- (2,2) -- (0,2) -- cycle;

	\node at (0.9,0.85) {\scriptsize \xmark};
	\node at (0.35,0.2) {\scriptsize \xmark};
	\node at (0.1,0.25) {\scriptsize \xmark};

	\draw[decorate,decoration={brace,mirror}] (0,-0.1) -- (1,-0.1) node[midway,below,font=\tiny] {$n$};
\end{tikzpicture}}}
	\caption{\label{fig:quad-search}
		A simple quadrant search algorithm, showing the first 5 probes.
		Only three of the four quadrants must be probed, as the POIs (denoted by \xmark)
		must be in the final quadrant if the first three probes fail.
		Regions with diagonal lines correspond to failed probes, while probe 5
		is a successful probe.
		The search will continue in the 5th region.
		For simplicity, figures represent $L_\infty$ probes, but our algorithms
		and figures translate directly to $L_1$ probes when considered
		diagonally.
	}
\end{figure}
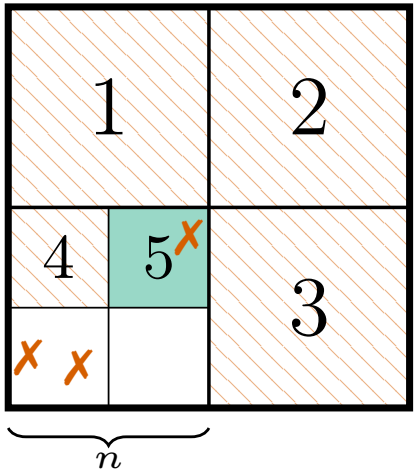

While our subsequent algorithms will reduce the total number of probes, this
algorithm behaves best with regards to the maximum number of times a POI
must respond to a probe, $R_\text{max}$.
Specifically, since a POI only responds at most once per layer for both 2D
and 3D, we have that $R_\text{max} \leq \lceil \log{n} \rceil$ responses.
We note the similarity between this algorithm and the hexagonal algorithms
of~\cite{gila2025combinatorialdronesearching}.

\subsection{Domino Algorithms}
Our quadrant algorithm was able to find a POI using at most
$3 \lceil \log{n} \rceil$ probes, and the question remains---can we do better?
For the 2D case, since we start with an area of $(2n)^2$ and end with an area of
no larger than $(2)^2$, and each probe, in the worst case, at most halves the
remaining area, there is a trivial lower bound of $2 \lceil \log{n} \rceil$
probes.
Similar reasoning can be used to lower bound the worst case number of probes in the
3D case to at least $3 \lceil \log{n} \rceil$ probes.
In this section, we introduce our first two algorithms, which are able to achieve
within constant factors of these lower bounds; due to their structure, we
refer to these as the 2D and 3D domino algorithms, respectively.

\subsubsection{The 2D Domino Algorithm}
%
In the 2D domino algorithm, we refer to a position as a 2-domino if it consists
of two equally sized areas---a $d \times d$ area where a POI is known
to exist, and an adjacent $d \times d$ area known to be empty, i.e., where
\textit{no} POI exists.
Let the empty region be to the left of the remaining search area as depicted in
\Cref{fig:2-domino}, without loss of generality.
Our first probe has radius $d / 2$, and is placed halfway between the two areas.
If the probe fails, we know that the POI must be in the remaining right half
of the search area.
If the probe succeeds, however, we take advantage of the fact that we know that
the left half of the probe is empty, and similarly reduce the search area.
Regardless of the result, we perform a second probe with radius $d / 4$ in one
half of the remaining search area, and achieve a new 2-domino with a quarter of
the area.
This is optimal, since the remaining area is halved with each probe.
\Cref{fig:2-domino} depicts this procedure.

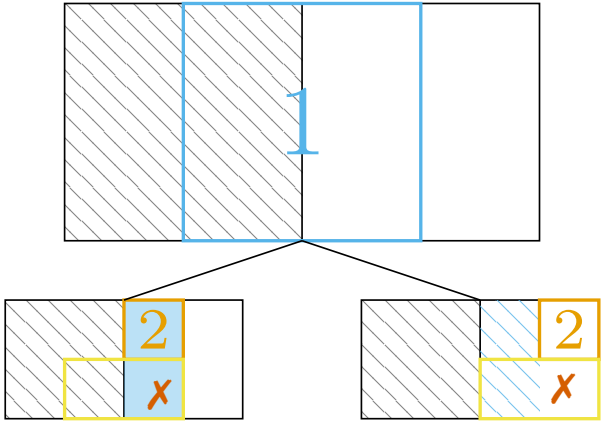
\begin{figure}[ht!]
	\centering
	\resizebox*{\linewidth}{!}{\begin{tikzpicture}
	\draw[line width=0.2] (-1,-0.5) rectangle (1,0.5);
	\draw[line width=0.2] (0,-0.5) -- (0,0.5);
	\fill[pattern={Lines[angle=-45,distance=6]}, opacity=0.7] (-1,-0.5) rectangle (0,0.5);

	\draw[okabe3] (-0.5,-0.5) rectangle (0.5,0.5);
	\node[okabe3] at (0,0) {\large 1};
	
	\begin{scope}[scale=0.5, shift={(-1.5,-2)}]
		\draw[line width=0.2] (-1,-0.5) rectangle (1,0.5);
		\draw[line width=0.2] (0,-0.5) -- (0,0.5);
		\fill[pattern={Lines[angle=-45,distance=6]}, opacity=0.7] (-1,-0.5) rectangle (0,0.5);
		\fill[okabe3, opacity=0.4] (0,-0.5) rectangle (0.5,0.5);

		\draw[okabe2] (0,0) rectangle (0.5,0.5);
		\node[okabe2] at (0.25,0.25) {\scriptsize 2};
		\node at (0.3,-0.3) {\tiny \xmark};

		\draw[okabe5] (-0.5,-0.5) rectangle (0.5,0);

		\node (a) at (0,0.5) {};
	\end{scope}

	\begin{scope}[scale=0.5, shift={(1.5,-2)}]
		\draw[line width=0.2] (-1,-0.5) rectangle (1,0.5);
		\draw[line width=0.2] (0,-0.5) -- (0,0.5);
		\fill[pattern={Lines[angle=-45,distance=6]}, opacity=0.7] (-1,-0.5) rectangle (0,0.5);
		\fill[pattern={Lines[angle=-45,distance=6]}, pattern color=okabe3, opacity=0.9] (0,-0.5) rectangle (0.5,0.5);

		\draw[okabe2] (0.5,0) rectangle (1,0.5);
		\node[okabe2] at (0.75,0.25) {\scriptsize 2};
		\node at (0.7,-0.25) {\tiny \xmark};

		\draw[okabe5] (0,-0.5) rectangle (1,0);

		\node (b) at (0,0.5) {};
	\end{scope}

	\draw[line width=0.2] (0,-0.5) -- (a.center);
	\draw[line width=0.2] (0,-0.5) -- (b.center);
\end{tikzpicture}}
	\caption{\label{fig:2-domino}
		The recursive 2-domino procedure as used in the 2D domino algorithm.
		Each probe reduces the remaining area by a factor of 2.
		Regardless of the results of these two probes, we are left with a new
		2-domino where each dimension is halved, depicted in yellow.
	}
\end{figure}

The question that remains, however, is how to achieve the initial 2-domino.
To this end, our 2D domino algorithm performs the top layer using the simple
quadrant algorithm depicted in \Cref{fig:quad-search}.
In the best case, if the first probe succeeds, we have reduced the area by a
factor of 4 and simply continue our algorithm recursively into this quadrant.
Otherwise, if the first probe fails, any subsequent probe that succeeds is
adjacent to a quadrant that is known to be empty, and we can initiate our
2-domino procedure.
In the worst case, it will take all 3 initial probes to reach a domino with a
quarter of the remaining search area, resulting in $ P(n) \leq 2 \lceil \log{n}
\rceil + 1$ probes.

\subsubsection{The 3D Domino Algorithm}
Unlike for our quadrant algorithm, it is not as straightforward to extend our 2D
domino algorithm to a 3D algorithm, and we only do so under the $L_\infty$
metric.
Our first step is to define a 4-domino as a 3D region consisting of four equally
sized $d \times d \times d$ cubes, where a POI is known to exist in one, and
all the rest are known to be empty.
As before, the first probe has radius $d / 2$, and
is placed halfway between the search area and one of
the two adjacent empty regions.
The second probe also has the same radius, and is placed orthogonally depending
on the result of the first probe as to halve the remaining volume again.
Finally, we perform the final probe of radius $d / 4$ in a half of the
remaining search area, resulting in a new 4-domino with an eighth of the
original volume.
See \Cref{fig:4-domino}.

\begin{figure}[ht!]
	\centering
	\def\probeone{1}
	\def\probetwo{1}
	\def\probethree{1}
	\def\showsubdominos{1}
	\resizebox*{\linewidth}{!}{\tdplotsetmaincoords{60}{-30}
\begin{tikzpicture}[tdplot_main_coords]



	\fill[pattern={Lines[angle=-45,distance=3]}, pattern color=okabe1, opacity=0.1] (0,0,0) -- (1,0,0) -- (1,2,0) -- (0,2,0) -- (0,0,0);
	\fill[pattern={Lines[angle=-45,distance=3]}, pattern color=okabe1, opacity=0.1] (2,1,0) -- (2,1,1) -- (2,2,1) -- (2,2,0) -- (2,1,0);
	\fill[pattern={Lines[angle=-45,distance=3]}, pattern color=okabe1, opacity=0.1] (0,2,0) -- (2,2,0) -- (2,2,1) -- (0,2,1) -- (0,2,0);
	\fill[pattern={Lines[angle=-45,distance=3]}, pattern color=okabe1, opacity=0.1] (0,1,0) -- (1,1,0) -- (1,1,1) -- (0,1,1) -- (0,1,0);
	\fill[pattern={Lines[angle=-45,distance=3]}, pattern color=okabe1, opacity=0.1] (1,0,0) -- (1,2,0) -- (1,2,1) -- (1,0,1) -- (1,0,0);
	\fill[pattern={Lines[angle=-45,distance=3]}, pattern color=okabe1, opacity=0.1] (1,1,0) -- (2,1,0) -- (2,2,0) -- (1,2,0) -- (1,1,0);


	\draw[dashed, line width=0.1, opacity=0.2] (2,2,1) -- (2,2,0) -- (2,0,0);
	\draw[dashed, line width=0.1, opacity=0.2] (2,2,0) -- (0,2,0);

	\draw[line width=0.1,opacity=0.8] (0,0,0) -- (2,0,0) -- (2,0,1) -- (2,2,1) -- (0,2,1) -- (0,0,1) -- (0,0,0) -- (0,2,0) -- (0,2,1);
	\draw[line width=0.1,opacity=0.8] (0,0,1) -- (2,0,1);

	\draw[dashed, line width=0.1, opacity=0.2] (1,0,0) -- (1,2,0) -- (1,2,1);
	\draw[dashed, line width=0.1, opacity=0.2] (0,1,0) -- (2,1,0) -- (2,1,1);
	\draw[dashed, line width=0.1, opacity=0.2] (1,1,0) -- (1,1,1);

	\draw[line width=0.1,opacity=0.8] (1,0,0) -- (1,0,1) -- (1,2,1);
	\draw[line width=0.1,opacity=0.8] (0,1,0) -- (0,1,1) -- (2,1,1);

	\ifdef{\probeone}{
		\draw[dashed, okabe2, line width=0.4, opacity=0.8] (1,0.5,0) -- (1,1.5,0) -- (2,1.5,0) -- (2,0.5,0) -- (1,0.5,0) -- (1,0.5,1) -- (2,0.5,1) -- (2,1.5,1) -- (1,1.5,1) -- (1,1.5,0);
		\draw[dashed, okabe2, line width=0.4, opacity=0.8] (2,0.5,0) -- (2,0.5,1);
		\draw[dashed, okabe2, line width=0.4, opacity=0.8] (2,1.5,0) -- (2,1.5,1);

		\fill[pattern={Lines[angle=-45,distance=3]}, pattern color=okabe2, opacity=0.2] (1,0.5,0) -- (2,0.5,0) -- (2,0.5,1) -- (1,0.5,1) -- cycle;

		\draw[okabe2, line width=1] (1,0.5,1) -- (1,1.5,1) -- (2,1.5,1) -- (2,0.5,1) -- cycle;
	}{}

	\ifdef{\probetwo}{
		\draw[dashed, okabe3, line width=0.4, opacity=0.8] (0.5,0,0) -- (0.5,1,0) -- (1.5,1,0) -- (1.5,0,0);
		\draw[dashed, okabe3, line width=0.4, opacity=0.8] (0.5,1,0) -- (0.5,1,1);
		\draw[dashed, okabe3, line width=0.4, opacity=0.8] (1.5,1,0) -- (1.5,1,1);

		\fill[pattern={Lines[angle=-45,distance=3]}, pattern color=okabe3, opacity=0.2] (1,0,0) -- (1.5,0,0) -- (1.5,0,1) -- (1,0,1) -- cycle;
		
		\draw[okabe3, line width=1] (0.5,0,0) -- (1.5,0,0) -- (1.5,0,1) -- (0.5,0,1) -- (0.5,1,1) -- (1.5,1,1) -- (1.5,0,1);
		\draw[okabe3, line width=1] (0.5,0,0) -- (0.5,0,1);
	}{}

	\ifdef{\probethree}{
		\draw[dashed, okabe4, line width=0.4, opacity=0.8] (1.5,0,0.5) -- (1.5,0.5,0.5) -- (2,0.5,0.5) -- (2,0,0.5);
		\draw[dashed, okabe4, line width=0.4, opacity=0.8] (1.5,0.5,0.5) -- (1.5,0.5,1);
		\draw[dashed, okabe4, line width=0.4, opacity=0.8] (2,0.5,0.5) -- (2,0.5,1);

		\fill[pattern={Lines[angle=-45,distance=3]}, pattern color=okabe4, opacity=0.2] (1.5,0,0.5) -- (2,0,0.5) -- (2,0,1) -- (1.5,0,1) -- cycle;
		\fill[pattern={Lines[angle=-45,distance=3]}, pattern color=okabe4, opacity=0.2] (1.5,0,1) -- (1.5,0.5,1) -- (2,0.5,1) -- (2,0,1) -- cycle;

		\draw[okabe4, line width=1] (1.5,0,0.5) -- (2,0,0.5) -- (2,0,1) -- (1.5,0,1) -- cycle;
		\draw[okabe4, line width=1] (1.5,0,1) -- (1.5,0.5,1) -- (2,0.5,1) -- (2,0,1);
	}{}

	\ifdef{\showsubdominos}{
		\draw[dashed, okabe5, line width=0.4, opacity=0.8] (1,0,0) -- (1,1,0) -- (2,1,0) -- (2,1,0.5) -- (1,1,0.5) -- (1,0,0.5);
		\draw[dashed, okabe5, line width=0.4, opacity=0.8] (2,0,0.5) -- (2,1,0.5);
		\draw[dashed, okabe5, line width=0.4, opacity=0.8] (1,1,0) -- (1,1,0.5);
		\draw[dashed, okabe5, line width=0.4, opacity=0.8] (1,0.5,0) -- (2,0.5,0) -- (2,0.5,0.5) -- (1,0.5,0.5) -- cycle;
		\draw[dashed, okabe5, line width=0.4, opacity=0.8] (1.5,0,0) -- (1.5,1,0) -- (1.5,1,0.5) -- (1.5,0,0.5);

		\draw[okabe5, line width=1] (1,0,0) -- (2,0,0) -- (2,0,0.5) -- (1,0,0.5) -- cycle;
		\draw[okabe5, line width=1] (1.5,0,0) -- (1.5,0,0.5);
	}{}

	\node at (1.8,0.15,0.1) {\tiny \xmark};

	\node at (1.88,0,1) {};
\end{tikzpicture}}
	\caption{\label{fig:4-domino}
		The recursive 4-domino procedure as used in the 3D domino algorithm.
		Each probe reduces the remaining volume by a factor of 2.
		Regardless of the results of these three probes, depicted in blue,
		orange, and green, respectively, we are left with a new 4-domino where each dimension
		is halved, depicted in yellow.
	}
\end{figure}
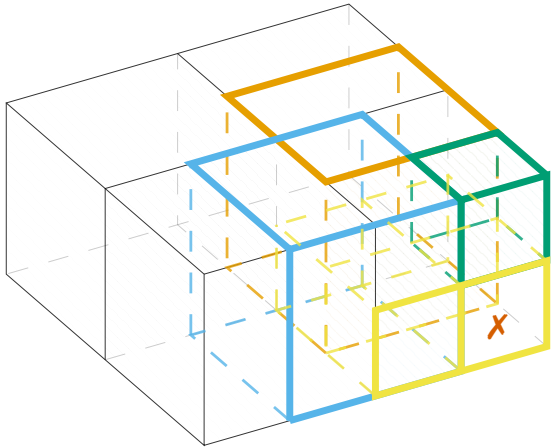

We have shown that once we reach a 4-domino shape, we are able to perform
optimally, reducing the remaining volume by a factor of 2 with every probe.
Unfortunately, there is no simple procedure by which we efficiently reach
a 4-domino shape from the original search area.
In \Cref{sec:3d-domino}, we show how by using
an additional construction, similar to the
2D 2-domino procedure, as an intermediate step, we are able to prove good results.
The resulting 3D domino algorithm requires
at most $3 \lceil \log{n} \rceil + 4$ probes in the worst case.


\subsection{The Central Binary Search Algorithm}
Our domino algorithms were able to achieve excellent probe complexities of $2
\lceil \log{n} \rceil + \mathcal{O}(1)$ in 2D and $3 \lceil \log{n} \rceil +
\mathcal{O}(1)$ in 3D.
These algorithms, however, make no attempt to minimize the distance traveled by
the search point, $\Delta$, which may be important
in practice 
in a real-world scenario.
In this section we discuss
a 2D algorithm
that not only minimizes the number of
probes, $P(n)$, but also minimizes the distance traveled by $\Delta$, $D(n)$.
More specifically, we will show how
our algorithm is
instance-optimal with
respect to distance, which we define as having $\Delta$ travel a distance of at
most $D(n) \in \mathcal{O}(\delta_\text{min})$, where $\delta_\text{min}$ is the
distance of the closest POI to the origin using either the $L_1$ or
$L_\infty$ metrics.
This algorithm performs two binary searches from the center of each dimension,
so we
refer to
it 
as the
2D
\textbf{central binary search}
algorithm.

Our algorithm can be thought
of in two distinct `phases', one for a binary search in 2D, and another for
a binary search in a `1D' edge, as depicted in
\Cref{fig:cbs-2d,fig:cbs-2d-1d}, respectively.
In a later section, we will show how this algorithm can be generalized not only
to three dimensions, but also beyond.

For the first phase, we perform a sequence of probes from the origin, without
moving $\Delta$, performing a binary search to find a width-1 shell containing
the nearest POI, as described in \Cref{alg:true-bin}.
This binary search takes at most $\lceil \log n \rceil$ probes.
Afterwards, we determine which one of the four edges of the shell contains a
POI, which can be performed with two additional probes, as depicted in
\Cref{fig:cbs-2d}.
While the probe does not need to move by much to perform these two probes, it
may need to move 1 or 2 units along the $x$ and $y$ axes, such as to only probe
a desired subset of the shell's edges.
We go into more detail about this in a later section.




\begin{algorithm}[hbt]
	\caption{Binary Search for the Initial Shell Radius}\label{alg:true-bin}
	\begin{algorithmic}[1]
		\State{\textbf{Output:} approximate distance $\tilde \delta$ to nearest
		POI, such that $\tilde \delta - 1 \leq \delta_\text{min} \leq \tilde \delta$.}

		\State $l \gets 0$, $h \gets n$ \Comment{lower and upper bounds}

		\While{$h - l > 1$}
			\State $m \gets \left \lfloor \frac{h + l}{2} \right \rfloor$
			\If{$p(0,0,m)$ succeeds}
				\State $h \gets m$ \Comment{POI is in the shell}
			\Else
				\State $l \gets m$ \Comment{POI is outside the shell}
			\EndIf
		\EndWhile

		\State \Return $h$
	\end{algorithmic}
\end{algorithm}

\begin{figure}[ht!]
	\centering
	\resizebox*{.8\linewidth}{!}{\trimbox{0 3 0 0}{

\begin{tikzpicture}
	\node at (1.36,1.67) {\tiny \xmark};
	\node at (1.7,0.35) {\tiny \xmark};
	\node at (0.8,0.1) {\tiny \xmark};

	\draw[line width=0.25, okabe4] (0.375,0.375) rectangle (1.625,1.625);
	\fill[pattern={Lines[angle=-45,distance=6]}, pattern color=okabe4, opacity=0.7] (0.375,0.375) rectangle (1.625,0.5);
	\fill[pattern={Lines[angle=-45,distance=6]}, pattern color=okabe4, opacity=0.7] (0.375,0.5) rectangle (0.5,1.625);
	\fill[pattern={Lines[angle=-45,distance=6]}, pattern color=okabe4, opacity=0.7] (1.5,0.5) rectangle (1.625,1.625);
	\fill[pattern={Lines[angle=-45,distance=6]}, pattern color=okabe4, opacity=0.7] (0.5,1.5) rectangle (1.5,1.625);

	\draw[line width=0.5, okabe2] (0.25,0.25) rectangle (1.75,1.75);

	\fill[pattern={Lines[angle=-45,distance=6]}, pattern color=okabe3, opacity=0.7] (0.5,0.5) rectangle (1.5,1.5);
	\draw[line width=1, okabe3] (0.5,0.5) rectangle (1.5,1.5);

	\draw[line width=1] (0,0) rectangle (2,2);

	\draw[densely dotted,color=okabe4,line width=0.25] (1.625,0.375) -- (1.625,-0.06) node[below, yshift=2pt] {\tiny 3};

	\draw[densely dotted,color=okabe2, line width=0.5] (1.75,0.25) -- (1.75,-0.08) node[below, yshift=2pt] {\tiny 2};
	\draw[densely dotted,color=okabe3,line width=1] (1.5,0.5) -- (1.5,-0.1) node[below, yshift=2pt] {\tiny 1};

	\draw[decorate,decoration={brace}] (1.625,1.625) -- (1.75,1.625) node[midway,above,font=\tiny] {1};

	\fill[okabe8] (1,1) circle (0.04);

	\begin{scope}[scale=0.75, shift={(2.65,1.35)}]
		\node at (1.36,1.67) {\tiny \xmark};
		\node at (1.7,0.35) {\tiny \xmark};

		\draw[line width=0.25, okabe4] (0.375,0.375) rectangle (1.625,1.625);
		\fill[pattern={Lines[angle=-45,distance=6]}, pattern color=okabe4, opacity=0.3] (0.375,0.375) rectangle (1.625,1.625);

		\draw[line width=0.5, okabe2] (0.25,0.25) rectangle (1.75,1.75);

		\draw[line width=0.75, okabe7, opacity=0.8] (0.375,0.375) rectangle (1.75,1.75);


		\node[okabe7] at (1.125,1.125) {\large \textbf{4}};

		\node (a) at (0.25,1) {};
		\node (b) at (1,0.25) {};
	\end{scope}

	\begin{scope}[scale=0.75, shift={(2.65,-0.65)}]
		\node at (1.36,1.67) {\tiny \xmark};
		\node at (1.7,0.35) {\tiny \xmark};

		\draw[line width=0.25, okabe4] (0.375,0.375) rectangle (1.625,1.625);
		\fill[pattern={Lines[angle=-45,distance=6]}, pattern color=okabe4, opacity=0.3] (0.375,0.375) rectangle (1.625,1.625);

		\draw[line width=0.5, okabe2] (0.25,0.25) rectangle (1.75,1.75);

		\draw[line width=0.75, okabe7, opacity=0.8] (0.375,0.375) rectangle (1.75,1.75);

		\fill[pattern={Lines[angle=-45,distance=6]}, pattern color=okabe3, opacity=0.7] (1.625,0.375) rectangle (1.75,1.625);
		\draw[line width=1, okabe3] (0.5,0.375) rectangle (1.75,1.625);


		\node[okabe3] at (1.1875,1.0625) {\large \textbf{5}};
		\node (c) at (1,1.75) {};
	\end{scope}

	\draw[line width=0.2, ->] (2,1) -- (a.center);
	\draw[line width=0.2, ->] (b.center) -- (c.center);
\end{tikzpicture}}}
	\caption{\label{fig:cbs-2d}
		The 2D central binary search algorithm.
		The probe is always centered at the origin, depicted by a purple point,
		but varies in size, increasing after every failed probe, and decreasing
		after every successful probe.
		It continues performing this binary search until reaching the closest
		width-1 shell to the origin where a POI is known to exist.
		After this, only two additional probes are needed to determine which of
		the four edges of the shell contains a POI.
	}
\end{figure}
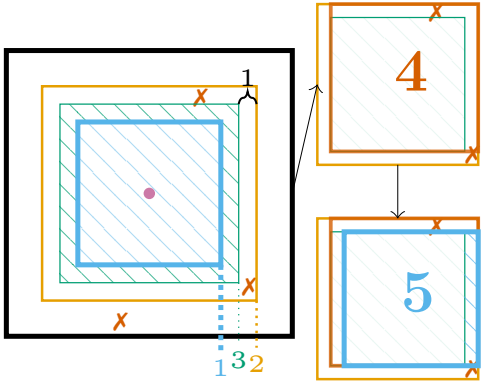

After this, we have reduced the problem to a 1D search along an edge
of the original shell.
A similar binary search is performed to find the two width-1 squares nearest to
the center of the edge which are known to contain a POI.
While this binary search also takes at most $\lceil \log n \rceil$ probes, it
may require $\Delta$ to move back and forth along the line from the origin and
the center of the edge.
See \Cref{fig:cbs-2d-1d}.
The first probe is performed at the center of this line, and requires the probe
to move $\tilde \delta/2$ units from the origin under both the $L_1$ and
$L_\infty$ metrics.
The second probe will be performed $\tilde \delta/4$ units away, either back
towards the origin in the case of a failed probe, or towards the edge in the
case of a successful probe.
In fact, each subsequent probe will require moving $\Delta$ half the distance
as the previous, and so the total distance traveled by $\Delta$ to perform this
search is at most $\tilde \delta$.
One final probe then determines which of the two width-1 squares
contains a POI, which again requires moving $\Delta$ only a small
constant distance.
And finally, in order to reach the location of the POI, $\Delta$ must
move a final distance of at most $\tilde \delta$.

\begin{figure}[ht!]
	\centering
	\resizebox*{\linewidth}{!}{\trimbox{0 0 0 0}{\begin{tikzpicture}
	\begin{scope}[]
		\node at (1.36,1.67) {\tiny \xmark};

		\fill[pattern={Lines[angle=-45,distance=6]}, opacity=0.3] (0.5,0.5) rectangle (1.75,1.625);

		\draw[line width=0.2] (0.5, 0.5) rectangle (1.75,1.75);
		\draw[line width=0.2] (0.5,1.625) -- (1.75,1.625);
		\draw[line width=0.3, densely dotted] (1.125,0.5) -- (1.125,1.625);

		\draw[line width=0.5, okabe3] (0.8125,1.125) rectangle (1.4375,1.75);

		\fill[okabe8] (1.125,1.4375) circle (0.04);

		\node (a) at (1.75,1.125) {};
	\end{scope}

	\begin{scope}[shift={(1.5,0)}]
		\node at (1.36,1.67) {\tiny \xmark};

		\fill[pattern={Lines[angle=-45,distance=6]}, opacity=0.3] (0.5,0.5) rectangle (1.75,1.625);

		\draw[line width=0.2] (0.5, 0.5) rectangle (1.75,1.75);
		\draw[line width=0.2] (0.5,1.625) -- (1.75,1.625);
		\draw[line width=0.3, densely dotted] (1.125,0.5) -- (1.125,1.625);
		
		\draw[line width=0.5, okabe3] (0.8125,1.125) rectangle (1.4375,1.75);
		
		\fill[pattern={Lines[angle=-45,distance=6]}, pattern color=okabe2, opacity=0.9] (0.96875,1.625) rectangle (1.28125,1.75);
		\draw[line width=0.5, okabe2] (0.96875,1.4375) rectangle (1.28125,1.75);

		\fill[okabe8] (1.125,1.59375) circle (0.04);


		\node (b) at (0.5,1.125) {};
	\end{scope}

	\draw[line width=0.2, ->] (a.center) -- (b.center);
\end{tikzpicture}}}
	\caption{\label{fig:cbs-2d-1d}
		The second phase of the 2D CBS algorithm,
		where a POI is known to exist on a width-1 edge.
		One more binary search is conducted, where the probe moves along a
		line perpendicular to the center of the edge.
		After this, one final probe is necessary to determine which of the two
		1-by-1 regions contains a POI.
	}
\end{figure}
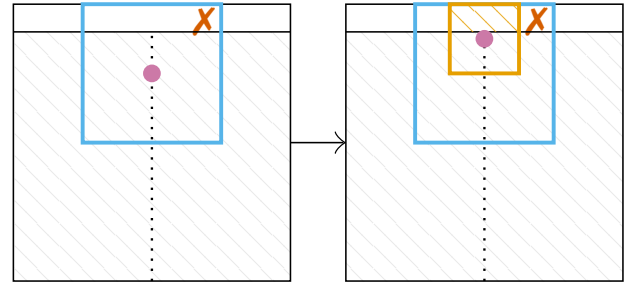

Our algorithm may require up to $2 \lceil \log n \rceil + 3$
probes to find a width-1 square containing a POI, but recall that our
objective is to reach a position within a \textit{distance} of 1 of our POI.
In truth, we can loosen our requirements to only finding a width-\textit{2}
square containing a POI, such that its center is at most 1 unit away
from the edges.
Each binary search will then take one fewer probe, resulting in a total of
$2 \lceil \log n \rceil + 1$ probes.
And as desired, the total distance traveled by $\Delta$ is at most
$2 \delta_\text{min} + \mathcal{O}(1)$, where $\delta_\text{min}$ is the distance to the nearest POI.

\section{Extended Search Strategies}\label{sec:general}

In this section, we expand on the key results of the previous section and show
how to generalize them to higher dimensions.

\subsection{Orthant Algorithm} \label{sec:orthant-algorithms}
As a warmup, we first extend our quadrant and octant algorithms to operate in
higher dimensions.
In this context, an \textbf{orthant} is the higher-dimensional analogue of a
quadrant (in 2D) or octant (in 3D), where each $k$-dimensional hypercube has
$2^k$ orthants.
By searching at most $2^k - 1$ orthants per layer, we can guarantee that we will
find a POI within at most $(2^k - 1) \lceil \log{n} \rceil$ probes,
where $k$ is the dimension of our search space,
while receiving at most $\lceil \log{n} \rceil$ responses.
To reduce the total distance traveled, we always probe adjacent orthants, which
can be done by following a Gray code, sometimes known as a single-digit
code~\cite{frank1953pulse}.
The distance between the centers of two adjacent orthants in the original
hypercube is $n$ under both the $L_1$ and $L_\infty$ metrics%
, so in the first layer alone we may travel a distance of over
$2^k n$.
Certainly such an algorithm should not be used when distance traveled is of any
concern.


\subsection{Generalized Central Binary Search (CBS) Algorithm} \label{sec:gen-cbs}
In this section we extend our 2D CBS algorithm to $k$ dimensions under the
$L_\infty$ metric.
Consider our true objective of a search strategy: to find a point in space that
is within distance 1 of a POI.
In $k$-dimensional space, this means a point with coordinates $x_1, \dots, x_k$,
where each coordinate is within distance 1 of the corresponding coordinate of a
POI.
As with the 2D case, we will split our algorithm into separate phases, where in
each phase we will fix the coordinates of at least one of the dimensions, thus
effectively reducing the dimension of our search space by 1.
Let phase $p$ refer to the phase in which our search space is effectively
confined to a $p$-dimensional subspace.
After at most $k$ phases, we will have fixed all $k$ coordinates, and
therefore we will have found a POI.
%
As in the 2D case, each phase $p$ starts by a binary search for the
(approximate) radius, $\tilde{\delta_p}$, of the smallest cube that contains at
least one POI.
Recall that each binary search takes at most $\lceil \log{n} \rceil$ probes.
While $\Delta$ is stationary in our $p$-dimensional subspace during each
binary search, for any phase $p < k$, $\Delta$ may need to move up to
$\tilde{\delta_p} < \delta_\text{min} + 1$ units overall (in the original
$k$-dimensional space) to conduct the search.
Regarding distance traveled, during the first phase, when $p = k$, $\Delta$ can
remain at the origin, while for all subsequent phases, $\Delta$ may need to move
up to $\tilde{\delta_p} < \delta_\text{min} + 1$ units.
This movement is necessary since although we visualize the search as oc
After this search, we are guaranteed that at least one POI contains a coordinate
that is within distance 1 of $\pm \tilde{\delta_p}$.

In $p$ dimensions, after the binary search determines our radius
$\tilde{\delta_p}$, we consider a $p$-cube with radius $\tilde{\delta_p}$, where
we know that the POI is located within distance 1 of one of the cubes' facets.
In our 2D algorithm, we were able to determine not only which coordinate to set
(the $x$ or the $y$), but also which sign to set ($\pm$) using only two probes.
The first probe eliminated two of the four edges, and the second probe
eliminated one of the two remaining edges.
See \Cref{fig:cbs-2d} (right).
Intuitively, it may seem possible to extend this idea to our $p$-cube,
halving the number of facets we consider with each probe.
Doing so would allow us determine the facet in $\log(2p)$ probes, but this is
unfortunately not possible in general.





One way to build this intuition is to consider the corners of our
$p$-cube.
The distance between any two corners is $2 \tilde{\delta_p}$, so any probe
that tests for the presence of a POI in 
two corners
concurrently must have a radius of at least $\tilde{\delta_p}$.
However, such a probe, when initiated from the center of the cube, will
encompass it entirely and not providing any new information%
, and when initiated from any other point, will
include regions outside of it which may contain other POIs.
Therefore, in the case of multiple possible POIs, we cannot probe
multiple corners of the cube concurrently.
This is unfortunate, since a $p$-cube has $2^p$ corners,
and since each of its facets contains $2^{p - 1}$ corners, if the POI is located
within distance 1 of a corner, we may need to probe $2^{p - 1}$ corners before
determining which facet contains the POI.
With this unfortunate observation in mind, not only do we need to individually
probe each facet, but our probes must be smaller than $\tilde{\delta_p}$, such
that each probe may not include the (lower dimensional) boundary of the facet.
In the case of $3$D cubes for example, each probe may not cover the edges nor
corners of the face being probed.
Therefore, we resort to not only individually probing each of the $2p$ facets,
but also the lower-dimensional faces, as described in
%
\Cref{alg:face-recurse}.


\begin{algorithm}[htb]
	\caption{Finding Cube Face near POI}\label{alg:face-recurse}
	\begin{algorithmic}[1]
		\State{\textbf{Output:} highest-dimensional $p$-cube face $f$ near
		at least one POI.}

		\For{$a = 1, \dots, p$}
			\For{$f \in (p - a)$-dimensional cube faces}
				\If{\Call{ProbeFace}{$f$}} \label{line:probe-face}
					\State \Return $f$ \Comment{POI is near face}
				\EndIf
			\EndFor
		\EndFor
	\end{algorithmic}
\end{algorithm}

The only remaining detail is the \texttt{ProbeFace($f$)} operation
for an arbitrary $(p-a)$-dimensional face $f$.

\begin{lemma} \label{lem:probe-face}
	Consider a $(p-a)$-dimensional face $f$ of a $p$-cube of radius
	$\tilde{\delta_p}$ centered at the origin.
	$f$ can be defined by a unique vector $s \in \{-1, 0, 1\}^p$ of length $p$,
	where `$a$' values are fixed to $\pm 1$, and the remaining values are 0,
	such that
	any point $x$ of $f$ satisfies $s \cdot x = a \tilde{\delta_p}$.
	\texttt{ProbeFace($f$)} can be performed by moving $\Delta$ to coordinate
	$s$ and conducting a probe with radius $\tilde{\delta_p} - 1$.
\end{lemma}

Using this lemma, which we prove in \Cref{sec:generalized-cbs}, we are able to probe each face
by moving $\Delta$ at most one unit from the origin under the $L_\infty$ metric
(since the maximum absolute value of $s$ is 1).
Since the number of faces required to probe is independent of $n$, we
obtain that the total number of probes required is at most
\begin{equation}
	k \lceil \log n \rceil + g(k),
\end{equation}
for some function $g(k)$ independent of $n$, and the total distance traveled is
\begin{equation}
	D(n) \leq k \cdot \delta_{\min} + 2g(k),
\end{equation}
with the factor of 2 accounting for the fact that $\Delta$ must return to the
origin after each face probe.
In other words, for any constant $k$, we get
\[
P(n) \leq k \lceil \log n \rceil + \mathcal{O}(1) \quad \text{and} \quad D(n) \leq k \cdot \delta_{\min} + \mathcal{O}(1).
\]
In the appendix, we show that $g(k) < 3^k$ and
discuss several realistic assumptions under which we can bound $g(k)$ to just
$k(k + 1)$, as long with experimental evidence supporting this going up to $k =
8$.
The biggest bottleneck for adapting CBS for $k > 2$ dimensions to the $L_1$
metric is that the faces of cross-polytopes are not, in general, other
cross-polytopes, although we conjecture that a similar algorithm can be
developed.

\subsection{Input-Sensitive Probe Complexity} \label{sec:input-sensitive-p}
There may be a scenario where the probes themselves are very expensive, and we
may not have a good estimate for the distance $\delta_\text{min}$ to the nearest
POI.
In this case, $\Delta$ can start by performing an exponential search from the
origin, increasing the probe radius by a factor of 2 per probe, until the first
successful probe.
Doing so requires at most $\lceil \log{\delta_\text{min}} \rceil$ probes,
and would limit the initial search area radius to $\delta_\text{min} \leq n < 2
\delta_\text{min}$, from which we can use our preferred algorithm to find the
POI.
For example, when using the central binary search algorithm, we achieve $P(n)
\leq (k + 1) \lceil \log{\delta_\text{min}} \rceil + \mathcal{O}(1)$.

\renewcommand{\emph}[1]{\textit{#1}}

\bibliographystyle{plainurl} 

\bibliography{refs}

%
%
\renewcommand{\emph}[1]{\textbf{\textit{#1}}}

\appendix

\section{Omitted Details and Results} \label{sec:omitted-results}

In this section, we provide more details and results omitted from the main
text.

\subsection{The 3D Domino Algorithm} \label{sec:3d-domino}
In the main text, we described how the 3D domino algorithm can indefinitely half
the remaining search volume for every probe once reaching a configuration we
refer to as a 4-domino.
In this configuration, we have four equally sized $d \times d \times d$ cubes,
where a POI is known to exist in one, and all the rest are known to be
empty, with the cubes arranged in a 2x2x1 grid.
See \Cref{fig:4-domino-simple}.
The key step glossed over in the main text is how to efficiently reach this
configuration.

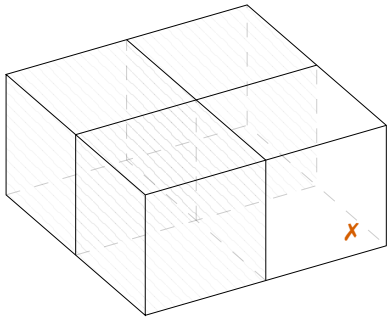
\begin{figure}[th!]
	\centering
	\resizebox*{0.7\linewidth}{!}{\tdplotsetmaincoords{60}{-30}
\begin{tikzpicture}[tdplot_main_coords]
	\fill[pattern={Lines[angle=-45,distance=3]}, pattern color=okabe1, opacity=0.2] (0,0,0) -- (1,0,0) -- (1,2,0) -- (0,2,0) -- cycle;
	\fill[pattern={Lines[angle=-45,distance=3]}, pattern color=okabe1, opacity=0.2] (2,1,0) -- (2,1,1) -- (2,2,1) -- (2,2,0) -- cycle;
	\fill[pattern={Lines[angle=-45,distance=3]}, pattern color=okabe1, opacity=0.2] (0,2,0) -- (2,2,0) -- (2,2,1) -- (0,2,1) -- cycle;
	\fill[pattern={Lines[angle=-45,distance=3]}, pattern color=okabe1, opacity=0.2] (0,1,0) -- (1,1,0) -- (1,1,1) -- (0,1,1) -- cycle;
	\fill[pattern={Lines[angle=-45,distance=3]}, pattern color=okabe1, opacity=0.2] (1,0,0) -- (1,2,0) -- (1,2,1) -- (1,0,1) -- cycle;
	\fill[pattern={Lines[angle=-45,distance=3]}, pattern color=okabe1, opacity=0.2] (1,1,0) -- (2,1,0) -- (2,2,0) -- (1,2,0) -- cycle;


	\draw[dashed, line width=0.1, opacity=0.2] (2,2,0) -- (2,2,1);
	\draw[dashed, line width=0.1, opacity=0.2] (0,2,0) -- (2,2,0) -- (2,0,0);

	\draw[dashed, line width=0.1, opacity=0.2] (1,0,0) -- (1,2,0) -- (1,2,1);
	\draw[dashed, line width=0.1, opacity=0.2] (0,1,0) -- (2,1,0) -- (2,1,1);
	\draw[dashed, line width=0.1, opacity=0.2] (1,1,0) -- (1,1,1);


	\draw[line width=0.1] (0,0,0) -- (2,0,0) -- (2,0,1) -- (2,2,1) -- (0,2,1) -- (0,0,1) -- (0,0,0) -- (0,2,0) -- (0,2,1);
	\draw[line width=0.1] (0,0,1) -- (2,0,1);

	\draw[line width=0.1] (1,0,0) -- (1,0,1) -- (1,2,1);
	\draw[line width=0.1] (0,1,0) -- (0,1,1) -- (2,1,1);
	
	\node at (1.8,0.15,0.1) {\tiny \xmark};
\end{tikzpicture}}
	\caption{\label{fig:4-domino-simple}
		The configuration referred to as a 4-domino.
		After reaching this configuration, we can indefinitely halve the
		remaining search volume with every probe.
	}
\end{figure}

Perhaps the most natural idea that comes to mind is to split the initial search
cube into its 8 equally sized octants, and then sequentially probe each of them
until we reach a 4-domino configuration.
However, since the 4-domino configuration relies on the fact that 3 regions are
known to be empty, it is possible that no such configuration exists, e.g., if
there exists a POI in every octant.
Recall that our objective with the 3D domino algorithm is to minimize the number
of probes to $3 \lceil \log{n} \rceil + \mathcal{O}(1)$, meaning that we must in
general reduce the search volume by a factor of 8 after every 3 probes.
Each octant has an eighth of the volume of the original search cube, and thus,
as long as one of the first 3 probes succeeds, we still succeed in reducing the
search volume by a factor of 8.
The case of a POI in every octant is therefore a very lucky case, since we are
able to reduce the search volume by a factor of 8 on every probe.
Our goal therefore, is to conduct the first 3 probes in such a way that even if
all three of them fail, we are guaranteed to reach a 4-domino configuration.
Unfortunately, this is not possible.
\Cref{fig:4-domino-fail} depicts one possible scenario.
In fact, no matter where the first 3 probes are conducted, POIs may be
placed adversarially such that after those probes no 4-domino configuration exists.

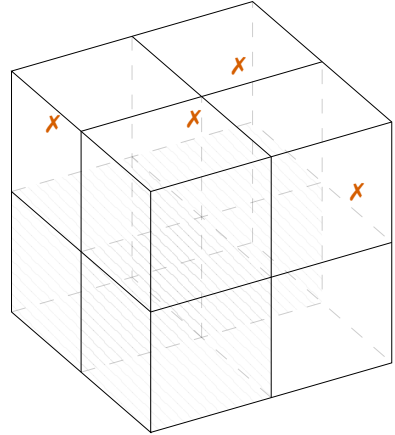
\begin{figure}[th!]
	\centering
	\resizebox*{0.7\linewidth}{!}{\tdplotsetmaincoords{60}{-30}
\begin{tikzpicture}[tdplot_main_coords]
	\begin{scope}[]
		\fill[pattern={Lines[angle=-45,distance=3]}, pattern color=okabe1, opacity=0.2] (0,0,0) -- (1,0,0) -- (1,2,0) -- (0,2,0) -- cycle;
		\fill[pattern={Lines[angle=-45,distance=3]}, pattern color=okabe1, opacity=0.2] (2,1,0) -- (2,1,1) -- (2,2,1) -- (2,2,0) -- cycle;
		\fill[pattern={Lines[angle=-45,distance=3]}, pattern color=okabe1, opacity=0.2] (0,2,0) -- (2,2,0) -- (2,2,1) -- (0,2,1) -- cycle;
		\fill[pattern={Lines[angle=-45,distance=3]}, pattern color=okabe1, opacity=0.2] (0,1,0) -- (1,1,0) -- (1,1,1) -- (0,1,1) -- cycle;
		\fill[pattern={Lines[angle=-45,distance=3]}, pattern color=okabe1, opacity=0.2] (1,0,0) -- (1,2,0) -- (1,2,1) -- (1,0,1) -- cycle;
		\fill[pattern={Lines[angle=-45,distance=3]}, pattern color=okabe1, opacity=0.2] (1,1,0) -- (2,1,0) -- (2,2,0) -- (1,2,0) -- cycle;
	

		\draw[dashed, line width=0.1, opacity=0.2] (2,2,0) -- (2,2,2);
		\draw[dashed, line width=0.1, opacity=0.2] (0,2,0) -- (2,2,0) -- (2,0,0);

		\draw[dashed, line width=0.1, opacity=0.2] (1,0,0) -- (1,2,0) -- (1,2,2);
		\draw[dashed, line width=0.1, opacity=0.2] (0,1,0) -- (2,1,0) -- (2,1,2);
		\draw[dashed, line width=0.1, opacity=0.2] (0,2,1) -- (2,2,1) -- (2,0,1);
		\draw[dashed, line width=0.1, opacity=0.2] (1,1,0) -- (1,1,2);
		\draw[dashed, line width=0.1, opacity=0.2] (1,0,1) -- (1,2,1);
		\draw[dashed, line width=0.1, opacity=0.2] (0,1,1) -- (2,1,1);


		\draw[line width=0.1] (0,0,0) -- (2,0,0) -- (2,0,2) -- (2,2,2) -- (0,2,2) -- (0,0,2) -- (0,0,0) -- (0,2,0) -- (0,2,2);
		\draw[line width=0.1] (0,0,2) -- (2,0,2);
	
		\draw[line width=0.1] (1,0,0) -- (1,0,2) -- (1,2,2);
		\draw[line width=0.1] (0,1,0) -- (0,1,2) -- (2,1,2);
		\draw[line width=0.1] (0,2,1) -- (0,0,1) -- (2,0,1);
	
		\node at (1.8,0.15,1.4) {\tiny \xmark};
		\node at (1.8,1.85,1.6) {\tiny \xmark};
		\node at (0.23,1.8,1.6) {\tiny \xmark};
		\node at (0.88,0.9,1.9) {\tiny \xmark};
	\end{scope}
\end{tikzpicture}}
	\caption{\label{fig:4-domino-fail}
		One possible configuration of the search cube after the first 3 probes
		of the 3D domino algorithm.
		Since there exists a POI in every octant in the top half of the cube,
		there does not exist any 4-domino configuration in this layer.
	}
\end{figure}

This poses a problem, since if we only discover a POI on our fourth probe and
continue the algorithm recursively from that octant, we have spent 4 probes to
reduce the search volume by a factor of 8, which, if done repeatedly would lead
to a probe complexity of $P(n) = 4 \lceil \log{n} \rceil$.
The 3D domino algorithm accounts for this by considering another intermediate
configuration, which is very similar to the 2-domino configuration of the 2D
domino algorithm.
In this configuration, we have two equally sized $d \times d \times d$ cubes,
where a POI is known to exist in one, and the other is known to be
empty.
From this configuration, while we are generally able to reduce the search volume
by a factor of 8 with every 3 probes, returning to another 2-domino
configuration.
However, if we get unlucky, these 3 probes in the 2-domino configuration may
fail to reduce the volume as desired---but the only way for this to occur would
induce a valid 4-domino configuration, from which point we can recurse
indefinitely.
Thus, such a failure can only occur once throughout the course of our algorithm.
See \Cref{fig:2-domino-3d} for a depiction of the first 3 probes in the 3D
2-domino procedure.
Note that this figure assumes that the first probe is successful without loss of
generality.
If either the second or third probe is successful, we can simply recurse into
another 2-domino configuration with an eighth of the volume.
Otherwise, we must just perform one more probe in either of the remaining two
octants to reach a 4-domino configuration.

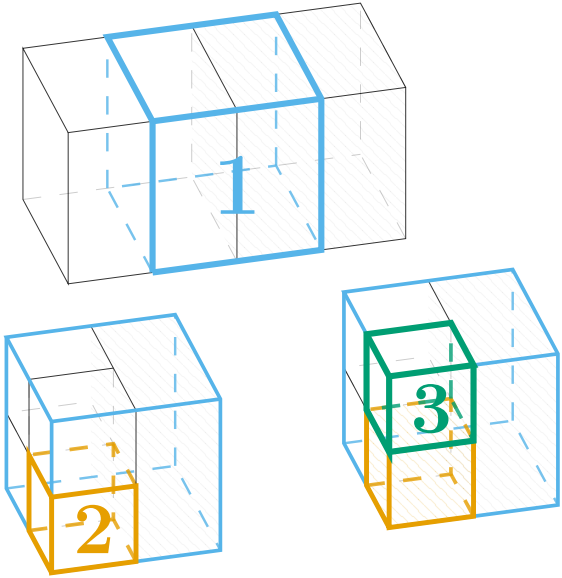
\begin{figure}[t!]
	\centering
	\resizebox*{\linewidth}{!}{\tdplotsetmaincoords{60}{-15}
\begin{tikzpicture}[tdplot_main_coords]
	\fill[pattern={Lines[angle=-45,distance=3]}, pattern color=okabe1, opacity=0.2] (1,0,1) -- (2,0,1) -- (2,1,1) -- (1,1,1) -- cycle;
	\fill[pattern={Lines[angle=-45,distance=3]}, pattern color=okabe1, opacity=0.2] (1,0,0) -- (2,0,0) -- (2,0,1) -- (1,0,1) -- cycle;
	\fill[pattern={Lines[angle=-45,distance=3]}, pattern color=okabe1, opacity=0.2] (1,0,0) -- (1,1,0) -- (1,1,1) -- (1,0,1) -- cycle;

	\draw[dashed, line width=0.1, opacity=0.2] (2,0,0) -- (2,1,0) -- (2,1,1) -- (2,0,1);
	\draw[dashed, line width=0.1, opacity=0.2] (0,1,0) -- (2,1,0);

	\draw[dashed, line width=0.1, opacity=0.2] (1,0,0) -- (1,1,0) -- (1,1,1);

	\draw[dashed, okabe3, line width=0.4, opacity=0.8] (0.5,0,0) -- (0.5,1,0) -- (1.5,1,0) -- (1.5,0,0);
	\draw[dashed, okabe3, line width=0.4, opacity=0.8] (0.5,1,0) -- (0.5,1,1);
	\draw[dashed, okabe3, line width=0.4, opacity=0.8] (1.5,1,0) -- (1.5,1,1);

	\draw[line width=0.1,opacity=0.8] (0,0,0) -- (2,0,0) -- (2,0,1) -- (0,0,1) -- (0,0,0) -- (0,1,0) -- (0,1,1) -- (0,0,1);
	\draw[line width=0.1,opacity=0.8] (0,1,1) -- (2,1,1) -- (2,0,1);
	
	\draw[line width=0.1,opacity=0.8] (1,0,0) -- (1,0,1) -- (1,1,1);

	\draw[okabe3, line width=1] (0.5,0,0) -- (1.5,0,0) -- (1.5,0,1) -- (1.5,1,1) -- (0.5,1,1) -- (0.5,0,1) -- (0.5,0,0);
	\draw[okabe3, line width=1] (0.5,0,1) -- (1.5,0,1);




	\node[okabe3] at (1,0,0.5) {\Large \textbf 1};



	\begin{scope}[shift={(-1,-1.5,-1)}]
		\fill[pattern={Lines[angle=-45,distance=3]}, pattern color=okabe1, opacity=0.2] (1,0,1) -- (1.5,0,1) -- (1.5,1,1) -- (1,1,1) -- cycle;
		\fill[pattern={Lines[angle=-45,distance=3]}, pattern color=okabe1, opacity=0.2] (1,0,0) -- (1.5,0,0) -- (1.5,0,1) -- (1,0,1) -- cycle;
		\fill[pattern={Lines[angle=-45,distance=3]}, pattern color=okabe1, opacity=0.2] (1,0,0) -- (1,1,0) -- (1,1,1) -- (1,0,1) -- cycle;

		\draw[dashed, line width=0.1, opacity=0.2] (0.5,0,0) -- (0.5,1,0) -- (0.5,1,1);
		\draw[dashed, line width=0.1, opacity=0.2] (0.5,0.5,0.5) -- (1,0.5,0.5) -- (1,1,0.5) -- (0.5,1,0.5) -- (0.5,0,0.5);
		\draw[dashed, line width=0.1, opacity=0.2] (1,0,0.5) -- (1,0.5,0.5);
		\draw[dashed, line width=0.1, opacity=0.2] (0.5,0.5,1) -- (0.5,0.5,0) -- (1,0.5,0) -- (1,0.5,1);

		\draw[dashed, okabe3, line width=0.4, opacity=0.8] (0.5,1,0) -- (1.5,1,0) -- (1.5,0,0);
		\draw[dashed, okabe3, line width=0.4, opacity=0.8] (1.5,1,0) -- (1.5,1,1);

		\draw[dashed, okabe2, line width=0.6, opacity=0.8] (1,0,0.5) -- (1,0.5,0.5) -- (0.5,0.5,0.5);
		\draw[dashed, okabe2, line width=0.6, opacity=0.8] (0.5,0.5,0) -- (1,0.5,0) -- (1,0.5,0.5);
		\draw[dashed, okabe2, line width=0.6, opacity=0.8] (1,0,0) -- (1,0.5,0);

		\draw[line width=0.1,opacity=0.8] (1,0,0) -- (1,0,1) -- (1,1,1);
		\draw[line width=0.1,opacity=0.8] (0.5,1,0.5) -- (0.5,0,0.5) -- (1,0,0.5);
		\draw[line width=0.1,opacity=0.8] (0.5,0.5,0) -- (0.5,0.5,1) -- (1,0.5,1);

		\draw[okabe3, line width=0.6] (0.5,0,0) -- (1.5,0,0) -- (1.5,0,1) -- (1.5,1,1) -- (0.5,1,1) -- (0.5,0,1) -- (0.5,0,0);
		\draw[okabe3, line width=0.6] (0.5,0,1) -- (1.5,0,1);
		\draw[okabe3, line width=0.6] (0.5,0,0) -- (0.5,1,0) -- (0.5,1,1);

		\draw[okabe2, line width=0.8] (0.5,0,0) -- (0.5,0.5,0) -- (0.5,0.5,0.5) -- (0.5,0,0.5) -- (1,0,0.5) -- (1,0,0) -- (0.5,0,0) -- (0.5,0,0.5);
	
		\node[okabe2] at (0.75,0,0.25) {\large \textbf 2};
	\end{scope}

	\begin{scope}[shift={(1,-1.5,-1)}]
		\fill[pattern={Lines[angle=-45,distance=3]}, pattern color=okabe1, opacity=0.2] (1,0,1) -- (1.5,0,1) -- (1.5,1,1) -- (1,1,1) -- cycle;
		\fill[pattern={Lines[angle=-45,distance=3]}, pattern color=okabe1, opacity=0.2] (1,0,0) -- (1.5,0,0) -- (1.5,0,1) -- (1,0,1) -- cycle;
		\fill[pattern={Lines[angle=-45,distance=3]}, pattern color=okabe1, opacity=0.2] (1,0,0) -- (1,1,0) -- (1,1,1) -- (1,0,1) -- cycle;

		\fill[pattern={Lines[angle=-45,distance=3]}, pattern color=okabe2, opacity=0.4] (0.5,0,0) -- (0.5,0.5,0) -- (0.5,0.5,0.5) -- (1,0.5,0.5) -- (1,0,0.5) -- (1,0,0) -- cycle;

		\draw[dashed, line width=0.1, opacity=0.2] (0.5,0,0) -- (0.5,1,0) -- (0.5,1,1);
		\draw[dashed, line width=0.1, opacity=0.2] (0.5,0.5,0.5) -- (1,0.5,0.5) -- (1,1,0.5) -- (0.5,1,0.5) -- (0.5,0,0.5);
		\draw[dashed, line width=0.1, opacity=0.2] (1,0,0.5) -- (1,0.5,0.5);
		\draw[dashed, line width=0.1, opacity=0.2] (0.5,0.5,1) -- (0.5,0.5,0) -- (1,0.5,0) -- (1,0.5,1);

		\draw[dashed, okabe3, line width=0.4, opacity=0.8] (0.5,1,0) -- (1.5,1,0) -- (1.5,0,0);
		\draw[dashed, okabe3, line width=0.4, opacity=0.8] (1.5,1,0) -- (1.5,1,1);

		\draw[dashed, okabe2, line width=0.6, opacity=0.8] (1,0,0.5) -- (1,0.5,0.5) -- (0.5,0.5,0.5);
		\draw[dashed, okabe2, line width=0.6, opacity=0.8] (0.5,0.5,0) -- (1,0.5,0) -- (1,0.5,0.5);
		\draw[dashed, okabe2, line width=0.6, opacity=0.8] (1,0,0) -- (1,0.5,0);

		\draw[dashed, okabe4, line width=0.6, opacity=0.8] (0.5,0,0.5) -- (1,0,0.5) -- (1,0.5,0.5) -- (0.5,0.5,0.5);
		\draw[dashed, okabe4, line width=0.6, opacity=0.8] (1,0,0.5) -- (1,0,1);
		\draw[dashed, okabe4, line width=0.6, opacity=0.8] (1,0.5,0.5) -- (1,0.5,1);

		\draw[line width=0.1,opacity=0.8] (1,0,0) -- (1,0,1) -- (1,1,1);
		\draw[line width=0.1,opacity=0.8] (0.5,1,0.5) -- (0.5,0,0.5) -- (1,0,0.5);
		\draw[line width=0.1,opacity=0.8] (0.5,0.5,0) -- (0.5,0.5,1) -- (1,0.5,1);

		\draw[okabe3, line width=0.6] (0.5,0,0) -- (1.5,0,0) -- (1.5,0,1) -- (1.5,1,1) -- (0.5,1,1) -- (0.5,0,1) -- (0.5,0,0);
		\draw[okabe3, line width=0.6] (0.5,0,1) -- (1.5,0,1);
		\draw[okabe3, line width=0.6] (0.5,0,0) -- (0.5,1,0) -- (0.5,1,1);

		\draw[okabe2, line width=0.8] (0.5,0,0) -- (0.5,0.5,0) -- (0.5,0.5,0.5) -- (0.5,0,0.5) -- (1,0,0.5) -- (1,0,0) -- (0.5,0,0) -- (0.5,0,0.5);
	
		\draw[okabe4, line width=1] (0.5,0,0.5) -- (0.5,0.5,0.5) -- (0.5,0.5,1) -- (1,0.5,1) -- (1,0,1) -- (0.5,0,1) -- cycle;
		\draw[okabe4, line width=1] (0.5,0,1) -- (0.5,0.5,1);
		\draw[okabe4, line width=1] (0.5,0,0.5) -- (1,0,0.5) -- (1,0,1);

		\node[okabe4] at (0.75,0,0.75) {\large \textbf 3};
	\end{scope}
\end{tikzpicture}}
	\caption{\label{fig:2-domino-3d}
		A depiction of the first 3 probes in the 3D 2-domino procedure.
		If all 3 probes fail, we are guaranteed to reach a 4-domino
		configuration.
	}
\end{figure}

Reaching this 2-domino configuration is easy, since, just like in the 2D domino
algorithm, we can simply probe adjacent octants.
Thus, our 3D domino algorithm begins by probing adjacent octants until one
succeeds.
If the first probe is successful, we simply recurse into that octant.
If a 4-domino configuration is reached, we can recurse indefinitely.
Finally, if a 4-domino configuration is not possible, since at least one probe
failed and we probe adjacent octants, we are guaranteed to reach a 2-domino
configuration, which we recurse into either indefinitely or until it becomes a
4-domino configuration.
Overall, the worst case scenario is where we probe 7 octants before determining
the location of the POI, after which we recurse optimally in a
4-domino configuration indefinitely.
Such a scenario would lead to a probe complexity of $P(n) = 7 + 3 \lceil \log{n
/ 2} \rceil = 3 \lceil \log{n} \rceil + 4$.

Reflecting on this algorithm, we see how, unlike the 2D domino algorithm which
required a single domino configuration which was easily reachable, the 3D domino
algorithm required a second intermediate configuration with a separate algorithm
on how to move from one configuration to the other.
This more tailored approach was not necessary for the generalized CBS algorithm,
which is broken up into simple, binary search-like phases, independent of the
search space dimension.

Nevertheless, it might be interesting to consider whether more domino-like
configurations could be used to extend our domino algorithms to higher
dimensions.
It is straightforward to see that for a $k$-dimensional search space, there
exists a $2^{k-1}$-domino configuration, where all but one of the $2^{k-1}$
cubes are known to be empty, after which point we can recurse indefinitely,
halving the search volume with every probe, such as the 2-domino configuration
in the 2D case and the 4-domino configuration in the 3D case.
Further, you can consider a $2^{k-2}$-domino configuration, a $2^{k-3}$-domino,
and so on, until reaching a 2-domino configuration.
%
%
%
While it is trivial to reach a 2-domino configuration from the initial hypercube
in any dimension, it is not clear whether a transition algorithm exists from a
2-domino to a 4-domino configuration for dimensions higher than 3.
In fact, we conjecture that this transition breaks down in dimensions higher than
3, ending such a generalization, but we leave this as an open question.

\subsection{The Generalized CBS Algorithm} \label{sec:generalized-cbs}
In the Generalized CBS algorithm, after the binary search of each phase which
determines an approximate remaining distance, $\tilde{\delta_p}$ to the nearest
POI, we must determine which of the $k$ coordinates to set to $\pm
\tilde{\delta_p}$.
In \Cref{alg:face-recurse}, we describe a procedure to do this by iteratively
probing smaller and smaller faces of our $p$-cube using a procedure we refer to
as \texttt{ProbeFace}, which relies on the following lemma.

\begin{lemma}(same as Lemma~\ref{lem:probe-face})
	Consider a $(p-a)$-dimensional face $f$ of a $p$-cube of radius
	$\tilde{\delta_p}$ centered at the origin.
	$f$ can be defined by a unique vector $s \in \{-1, 0, 1\}^p$ of length $p$,
	where `$a$' values are fixed to $\pm 1$, and the remaining values are 0,
	such that
	any point $x$ of $f$ satisfies $s \cdot x = a \tilde{\delta_p}$.
	\texttt{ProbeFace($f$)} can be performed by moving $\Delta$ to coordinate
	$s$ and conducting a probe with radius $\tilde{\delta_p} - 1$.
\end{lemma}

\begin{proof}
	Let us briefly consider which qualities we require from our probe.
	\begin{enumerate}
		\item The probe must not include any regions outside of our $p$-cube.
		\item The probe should include the entire face $f$, besides at most 1
		unit of padding from its boundary which will be covered by subsequent
		probes of lower-dimensional faces.
		\item The probe for a $(p-a)$-dimensional face $f$ should not include
		any other regions that are not already known to be empty.
	\end{enumerate}

	We prove the first property by contradiction. 
	Assume, for the sake of contradiction, that the probe includes a point $x$
	that is outside the $p$-cube of radius $\tilde{\delta_p}$ centered at the
	origin.
	A point $x$ is outside this $p$-cube if, for at least one coordinate $i$,
	its absolute value $|x_i|$ is greater than $\tilde{\delta_p}$.
	Without loss of generality, let us assume $x_i > \tilde{\delta_p}$ for this
	specific coordinate $i$.
	The probe is centered at $s$ (where $s_j \in \{-1, 0, 1\}$, so $s_j \leq 1$
	for all coordinates $j$) and has a radius of $\tilde{\delta_p} - 1$.
	For a point $x$ to be included in this probe (which is itself a $p$-cube),
	it must satisfy $|x_j - s_j| \leq \tilde{\delta_p} - 1$ for all coordinates
	$j$.

	Focusing on our specific coordinate $i$:
	We have $x_i > \tilde{\delta_p}$.
	We also know that $s_i \leq 1$.
	Consider the difference $x_i - s_i$.
	Since $x_i > \tilde{\delta_p}$ it follows that:
	$x_i - s_i > \tilde{\delta_p} - s_i$.
	Given $s_i \leq 1$, we have $\tilde{\delta_p} - s_i \geq \tilde{\delta_p} - 1$.
	Therefore, $x_i - s_i > \tilde{\delta_p} - 1$ and consequently $|x_i - s_i|
	> \tilde{\delta_p} - 1$.
	This result directly contradicts the condition for $x$ to be inside the
	probe.

	We now turn to the second property.
	The probe $P$ is centered at the coordinate $s$ and has a radius of
	$\tilde{\delta_p} - 1$ in each dimension.
	Thus, a point $y$ is within this probe if it satisfies $|y_j - s_j| \leq
	\tilde{\delta_p} - 1$ for all coordinates $j=1, \dots, p$.

	Let $I_{\text{free}} = \{j \mid s_j = 0\}$ be the set of indices for
	coordinates where $s_j$ is zero.
	These are the ``free'' coordinates along which the face $f$ extends.
	For any $x \in f$, $x_j \in [-\tilde{\delta_p}, \tilde{\delta_p}]$ for $j
	\in I_{\text{free}}$.
	Similarly, let $I_{\text{fixed}} = \{j \mid s_j \in \{-1, 1\}\}$ be the set of indices
	for coordinates where $s_j$ is non-zero.
	These are the ``fixed'' coordinates. For any point $x \in f$, its $j$-th
	coordinate is determined by $s_j$: $x_j = s_j \tilde{\delta_p}$ if $j \in
	I_{\text{fixed}}$.

	Let us consider an arbitrary point $x \in f$.
	We wish to show that $x$ satisfies the condition $|x_j - s_j| \leq
	\tilde{\delta_p} - 1$ for each coordinate $j$, considering two cases:
	\begin{itemize}
		\item If $j \in I_{\text{fixed}}$, then
			$x_j = s_j \tilde{\delta_p}$ for any point $x \in f$.
			We examine the condition for the probe:
			$|x_j - s_j| = |s_j \tilde{\delta_p} - s_j| = |s_j (\tilde{\delta_p} -
			1)|$.
			Since $s_j \in \{-1, 1\}$, we have $|s_j| = 1$.
			Therefore, $|s_j (\tilde{\delta_p} - 1)| = |\tilde{\delta_p} - 1|$.
			Assuming $\tilde{\delta_p} \geq 1$ (so that the probe radius
			$\tilde{\delta_p} - 1$ is non-negative), we have $|\tilde{\delta_p} - 1|
			= \tilde{\delta_p} - 1$.
			Thus, for all $j \in I_{\text{fixed}}$, the condition $|x_j - s_j| \leq
			\tilde{\delta_p} - 1$ is satisfied.
		
		\item If $j \in I_{\text{free}}$, then $s_j = 0$.
			For a point $x \in f$, its $j$-th coordinate $x_j$ can range within
			$[-\tilde{\delta_p}, \tilde{\delta_p}]$.
			The condition for $x$ to be included in the probe $P$ with respect
			to this $j$-th coordinate is $|x_j - s_j| \leq \tilde{\delta_p} -
			1$.
			Since $s_j = 0$, this simplifies to $|x_j| \leq \tilde{\delta_p} -
			1$.
			This means that for coordinates $j \in I_{\text{free}}$, the probe
			$P$ includes points $x \in f$ if their $j$-th coordinate $x_j$ lies
			in the interval $[-(\tilde{\delta_p} - 1), \tilde{\delta_p} - 1]$.
			While this is notably \textit{not} the entire range of $x_j$ (which
			is $[-\tilde{\delta_p}, \tilde{\delta_p}]$), it includes everything
			besides a 1-unit ``padding'' from the boundary of the face $f$ along
			this coordinate's axis, satisfying the second property.
	\end{itemize}

	Finally, we turn to the third property.
	We know from our algorithm, \Cref{alg:face-recurse}, that all
	higher-dimensional faces have been probed and are known to be empty.
	Further, we know that the internal region (besides a 1-unit padding) of the
	$p$-cube is empty.
	Thus, the probe must only avoid probing lower or equal-dimensional faces.
	Consider one such $(p - a')$-dimensional face $f'$ where $a' \geq a$, which
	is defined by a vector $s' \in \{-1, 0, 1\}^p$.
	Since $f'$ has at least $a'$ coordinates fixed to $\pm 1$, it must be the
	case that $s'$ has at least one coordinate $j$ such that $s'_j = \pm 1$ and
	$s_j = 0$.
	This means that for all points in $f'$, the $j$-th coordinate is fixed to
	$s'_j \tilde{\delta_p}$, i.e., that for any point $x' \in f'$, we have
	$|x_j| = \tilde{\delta_p}$.
	However, since our probe only has radius $\tilde{\delta_p} - 1$, and since
	$s_j = 0$, we have $|x'_j - s_j| = |x'_j| = \tilde{\delta_p} >
	\tilde{\delta_p} - 1$.
	Thus, our probe does not include any points in $f'$.
\end{proof}

We can now bound the total number, $g(k)$, of \texttt{ProbeFace} calls made
throughout the course of the algorithm.
A $p$-cube contains $3^p - 1$ faces. Thus, the total number of faces we must probe is technically $3^p - 2$ since we can skip the final probe.

However, since we start from higher-dimensional faces, if we take more than $2p$
probes, we are guaranteed to find a POI in a face which is lower by at least 2
dimensions, reducing $p$ by 2 instead of just 1.
It is straightforward to see that, in the worst case, a POI will only be located
in the final probe in the initial, $k$-dimensional cube, requiring $3^k - 2$
probes.
While this result is exceedingly unlikely, and is not possible if we require
that all of the POIs' coordinates differ in magnitude by at least 1, it is
nevertheless possible under our assumptions, and thus we only bound
$g(k)$ by $3^k - 2$.

As alluded to above, however, if we do require that all of the POIs' coordinates
differ in magnitude by at least 1, we know that the POIs will always be located
in a facet of our $p$-cube in each phase, requiring at most $2p$ probes.
Over the course of the algorithm, our total number of probes will therefore be
bounded by
\begin{equation*}
	\sum_{p=1}^{k} 2p = k(k+1),
\end{equation*}
thus bounded $g(k) = \mathcal{O}(k^2)$.

Alternatively, if we only consider cases where $k$ is small, i.e., $k =
o(\log{\log{n}})$, then we know that the total number \texttt{ProbeFace} calls
made for any dimension $p$, $3^p$, is bounded by $3^{\log{\log{n}}} =
o(\log{n})$.
However, if we take more than $2p$ probes, decreasing $p$ by 2, we can charge
our extra probes to the binary search of the $p - 1$ face we skipped, which we
expected to take $\lceil \log{n} \rceil$ probes.
Thus our worst case is the same as our previous example, and we similarly obtain
that $g(k) = \mathcal{O}(k^2)$.

\begin{table*}[tb!]
\centering
\begin{tabular}{|c|rrrr|rrrr|rrrr|}
\hline
& \multicolumn{4}{c|}{Domino Algorithms} & \multicolumn{4}{c|}{Orthant Algorithm} & \multicolumn{4}{c|}{Generalized CBS Algorithm} \\
$k$ & $\sigma$ & Avg & Max & Bound & $\sigma$ & Avg & Max & Bound & $\sigma$ & Avg & Max & Bound \\
\hline
1D & --- & --- & --- & --- & \textbf{0.00} & \textbf{0.95} & \textbf{0.95} & \textbf{1.00} & 0.00 & \textbf{0.95} & 1.00 & \textbf{1.00} \\
2D & 0.04 & \textbf{1.92} & \textbf{1.95} & \textbf{2.05} & 0.18 & 2.14 & 2.85 & 3.00 & \textbf{0.03} & 1.93 & 2.00 & 2.15/2.25 \\
3D & 0.11 & \textbf{2.92} & \textbf{3.05} & \textbf{3.20} & 0.46 & 4.16 & 6.35 & 7.00 & \textbf{0.07} & 2.96 & 3.10 & 3.40/4.10 \\
4D & --- & --- & --- & --- & 0.98 & 8.02 & 13.2 & 15.0 & \textbf{0.10} & \textbf{4.00} & \textbf{4.25} & \textbf{4.75/7.75} \\
5D & --- & --- & --- & --- & 2.00 & 15.6 & 26.1 & 31.0 & \textbf{0.14} & \textbf{5.06} & \textbf{5.40} & \textbf{6.20/16.8} \\
6D & --- & --- & --- & --- & 4.02 & 30.9 & 50.6 & 63.0 & \textbf{0.17} & \textbf{6.15} & \textbf{6.65} & \textbf{7.75/42.0} \\
7D & --- & --- & --- & --- & 8.05 & 61.3 & 103 & 127 & \textbf{0.21} & \textbf{7.26} & \textbf{7.95} & \textbf{9.40/116} \\
8D & --- & --- & --- & --- & 16.1 & 122 & 209 & \textbf{255} & \textbf{0.25} & \textbf{8.40} & \textbf{9.30} & 11.2/336 \\
\hline
\end{tabular}
\caption{Normalized number of probes ($P/\log n$) for different search algorithms across dimensions. The best (lowest) values are highlighted in bold.}
\label{tab:algorithm_metrics_P}
\end{table*}

\section{Experimental Results} \label{sec:experiments}

In this section we compare how each algorithm performs experimentally in terms
of our three metrics: the number of probes made, the distance traveled, and the
number of POI responses.
We compared results between the domino algorithms (in 2D and 3D), the orthant
algorithm (in 1D -- 8D), and the generalized CBS algorithm (also in 1D -- 8D).
Each algorithm was executed 60 million times, where $n = 2^{20}$.
The POIs were placed at uniformly random locations in the search area, namely
such that each of their coordinates was uniformly random in the range $[0, n]$.
These experiments focused on the $L_\infty$ metric.

It is clear from our results that both the domino and CBS algorithms
significantly outperform the generalized orthant algorithms for any dimension $k
> 1$, where in 1D space all algorithms are equivalent.
It is also possible to see that, as the dimension $k$ increases, the normalized
number of probes for the CBS algorithm gradually increases, as predicted by the
$g(k)$ dependence in its probe complexity.
See \Cref{fig:P}.

\begin{figure}[htb!]
	\centering
	\includegraphics[width=0.8\linewidth]{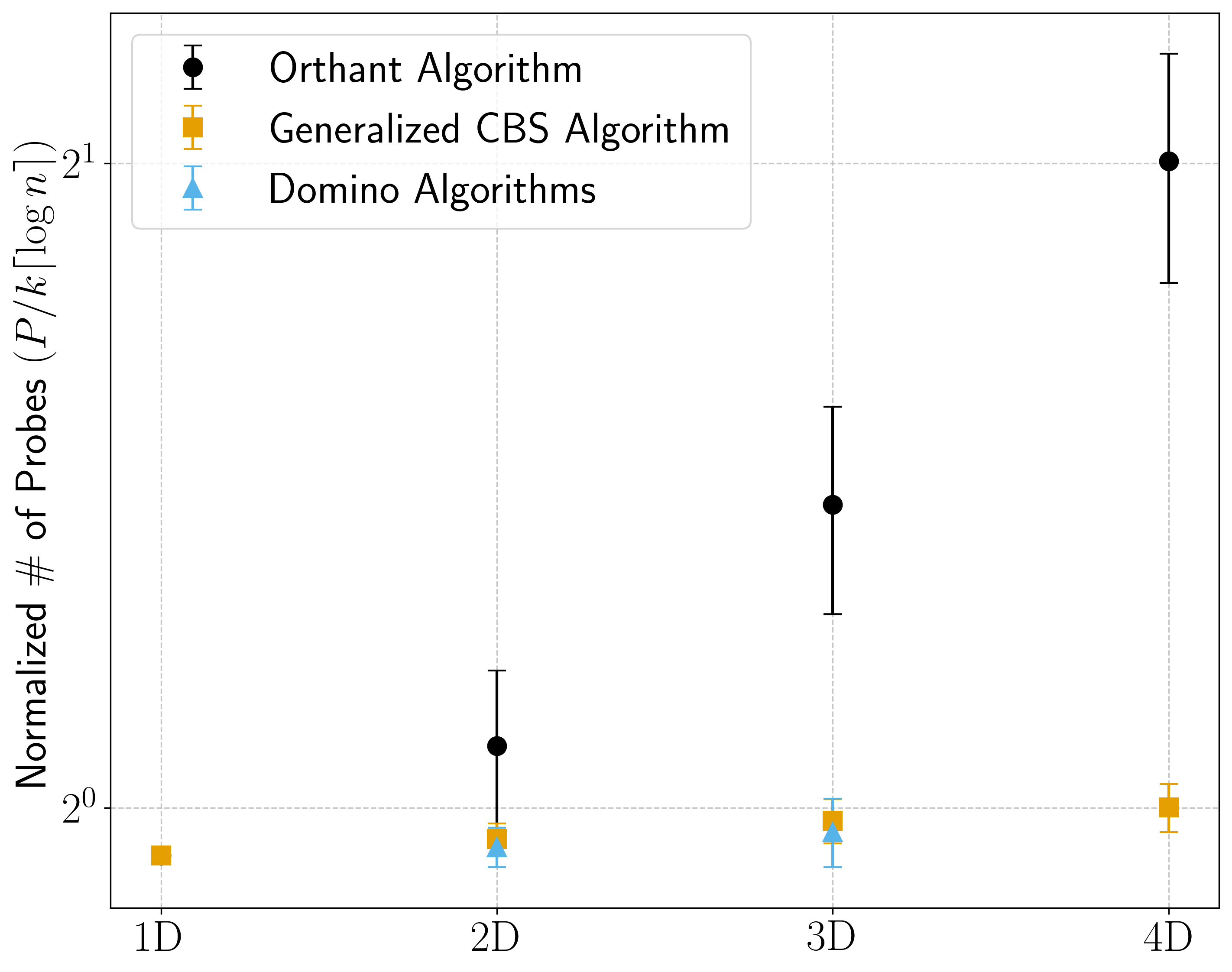}
	\caption{
		Simulation results for $P / k \lceil \log{n} \rceil$%
		.
		Error bars represent one standard deviation from the mean.
	}
	\label{fig:P}
\end{figure}

Regarding the distance traveled by the search point,
$\Delta$, we see that not only does the CBS
algorithm, being instance-optimal with regards to distance traveled, travel
significantly less distance with respect to the nearest POI's distance
on average, but it also has a much lower variability in the distance traveled
than the other algorithms.
See \Cref{fig:D}.

\begin{figure}[htb!]
	\centering
	\includegraphics[width=0.8\linewidth]{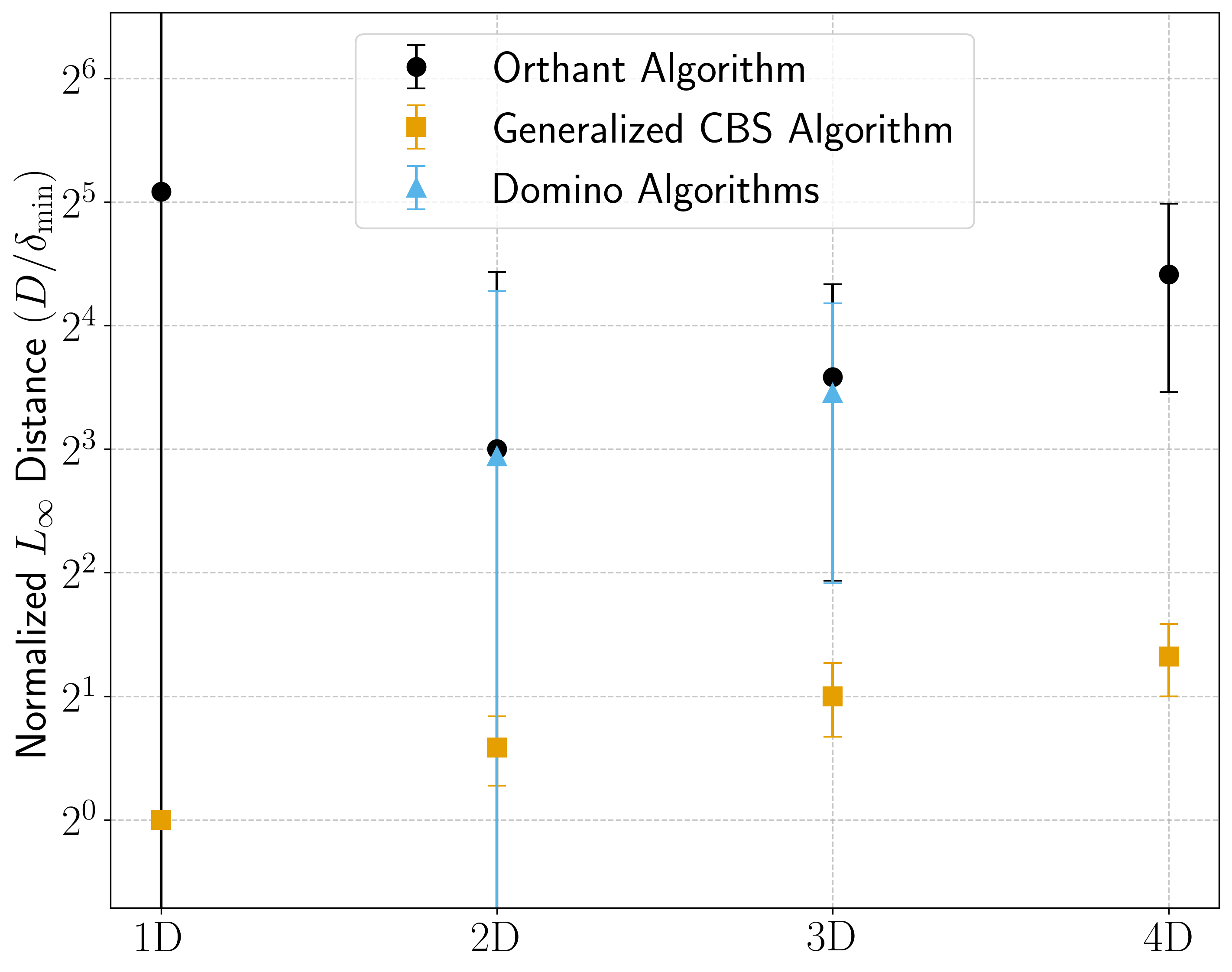}
	\caption{
		Simulation results for $D / \delta_\text{min}$%
		,
		where $n = 2^{20}$.
	}
	\label{fig:D}
\end{figure}



\begin{table*}[tb!]
\centering
\begin{tabular}{|c|rrrr|rrrr|rrrr|}
\hline
& \multicolumn{4}{c|}{Domino Algorithms} & \multicolumn{4}{c|}{Orthant Algorithm} & \multicolumn{4}{c|}{Generalized CBS Algorithm} \\
$k$ & $\sigma$ & Avg & Max & Bound & $\sigma$ & Avg & Max & Bound & $\sigma$ & Avg & Max & Bound \\
\hline
1D & --- & --- & --- & --- & $\sim 10^{4}$ & 27.8 & $\sim 10^{7}$ & $\sim 10^{6}$ & \textbf{0.00} & \textbf{1.00} & \textbf{1.00} & \textbf{1.00} \\
2D & 17.2 & 7.68 & $\sim 10^{4}$ & $\sim 10^{6}$ & 17.1 & 8.00 & $\sim 10^{4}$ & $\sim 10^{6}$ & \textbf{0.29} & \textbf{1.50} & \textbf{2.00} & \textbf{2.00} \\
3D & 7.17 & 11.0 & $\sim 10^{3}$ & $\sim 10^{7}$ & 8.14 & 12.0 & $\sim 10^{3}$ & $\sim 10^{7}$ & \textbf{0.41} & \textbf{2.00} & \textbf{3.00} & \textbf{3.00} \\
4D & --- & --- & --- & --- & 10.4 & 21.3 & $\sim 10^{3}$ & $\sim 10^{7}$ & \textbf{0.50} & \textbf{2.50} & \textbf{3.99} & \textbf{4.00} \\
5D & --- & --- & --- & --- & 16.7 & 40.0 & $\sim 10^{3}$ & $\sim 10^{7}$ & \textbf{0.58} & \textbf{3.00} & \textbf{4.97} & \textbf{5.00} \\
6D & --- & --- & --- & --- & 29.7 & 76.8 & $\sim 10^{3}$ & $\sim 10^{8}$ & \textbf{0.65} & \textbf{3.50} & \textbf{5.94} & \textbf{6.00} \\
7D & --- & --- & --- & --- & 55.2 & 149 & $\sim 10^{3}$ & $\sim 10^{8}$ & \textbf{0.71} & \textbf{4.00} & \textbf{6.86} & \textbf{7.00} \\
8D & --- & --- & --- & --- & 105 & 293 & $\sim 10^{3}$ & $\sim 10^{8}$ & \textbf{0.76} & \textbf{4.50} & \textbf{7.76} & \textbf{8.00} \\
\hline
\end{tabular}
\caption{Normalized $L_\infty$ distance ($D/\delta_\text{min}$) for different search algorithms across dimensions.}
\label{tab:algorithm_metrics_D}
\end{table*}

\subsection{Number of Probes Made (\texorpdfstring{$P$}{P})}
In \Cref{tab:algorithm_metrics_P} we compare the number of probes made by each
algorithm operating in different dimensions.
Each result is normalized by $\log n$, such that the theoretical lower bound for
any algorithm after the normalization is just $k$ for the $k$-dimensional case.
As expected, the domino algorithms, specifically designed to minimize the number
of probes, outperform the other algorithms in 2D and 3D across the average,
maximum, and theoretical bounds.
Interestingly, the generalized central binary search (CBS) algorithm, despite
performing slightly worse than the domino algorithms, performed slightly more
consistently, with a smaller standard deviation than the domino algorithms.
The orthant algorithms, as expected, perform very poorly in this metric, with
their average appearing to be, as expected, roughly $2^k / 2$.
It should be noted that the generalized CBS algorithm may theoretically only
find a POI when probing the final face it queries in the $k$-dimensional
space, recalling that there may be up to $g(k) = 3^k - 1$ faces.
In theory, this may, in 8 dimensions and beyond, lead to a worse performance
than the orthant algorithm.
However, recall that only the first $2k$ faces are facets, with all the others
being progressively lower dimensional.
These first $2k$ facets in general will be much much larger than all the
subsequent faces, meaning that, unless a POI is placed adversarially, it is
most probable that the POI will be found in one of the first $2k$ faces.
In fact, we can estimate this probability numerically.

Consider the case where we have a $k$-dimensional width-1 shell with radius
$\delta_\text{min}$, with only one POI placed uniformly at random in the
shell.
We probe each facet with a radius $\delta_\text{min} - 1$, covering a volume
contained by that facet of $2^{k-1} (\delta_\text{min} - 1)^{k - 1}$.
Overall, the $2k$ facet probes cover a volume of $k 2^k (\delta_\text{min}
- 1)^{k - 1}$.
The total volume of the $k$-dimensional shell is the total volume of the
hypercube, $2^k \delta_\text{min}^k$, minus the volume of the inner shell, which is
$2^k (\delta_\text{min} - 1)^k$.
Together, we get that the ratio of these volumes is
\begin{equation*}
	\frac{k (\delta_\text{min} - 1)^{k - 1}}{\delta_\text{min}^k - (\delta_\text{min} - 1)^k}.
\end{equation*}

Assuming that $\delta_\text{min}$ is at least $k$, this ratio is minimized as $k
\to \infty$, where it approaches $\frac{1}{e - 1} \approx 0.58$, where $e$ is
Euler's number.
In other words, assuming that a POI is placed uniformly at random, and that
the search area radius is only moderately larger than its dimension, it is most
probable that the POI will be found in one of the first $2k$ facets.
Surely this is the case for our experiments where the search area radius $n =
2^{20}$ is much larger than the dimension $k \leq 8$.
As such, we included not only the theoretical bounds for the generalized CBS
algorithm, but also the bounds assuming we always find the POI in one of the
facets we probe during each phase.
Our experimental results confirm that, despite the millions of simulations
performed, this bound was never violated.
This further supports our claim that under reasonable conditions, $g(k)$ can be
more accurately bounded by $k(k + 1)$.

\subsection{Distance Traveled (\texorpdfstring{$D$}{D})}
In \Cref{tab:algorithm_metrics_D} we compare the distance traveled by $\Delta$
for each algorithm operating in different dimensions.
Each result is normalized by $\delta_\text{min}$, such that the theoretical
lower bound for any algorithm after the normalization is just $1$.
We expect any \textit{instance-optimal} algorithm, with respect to the distance
traveled, to consistently travel within a constant multiple of this distance.
Our generalized CBS algorithm expects to travel a distance of at most $k$ times
the minimum distance, for example.
On the other hand, non-instance-optimal algorithms, such as the orthant and
domino algorithms, expect to travel a distance more dependent on the area of the
search area, $n$, and the dimension, $k$, essentially ignoring the position of
the nearest POI.
In the worst case, the POI is placed directly at the origin, but since we
disallowed this in our experiments, they would be placed at distance $1$ from
the origin, thus maximizing the ratio of the distance traveled to the minimum
distance.
While this extreme case evidently does not occur in our experiments, it is clear
from our results how much better the generalized CBS algorithm performs.
It outperforms the orthant and domino algorithms on all metrics and across all
dimensions, often by several orders of magnitude.
Reassuringly, the generalized CBS algorithm never travels a distance greater
than $k$ times the minimum distance, and seems to on average travel a distance
of $(k + 1) / 2$ times greater, with a moderate standard deviation of roughly
$20\%$ of the average.
It is worth noting that the search point, $\Delta$, 
in domino algorithms appears to travel
marginally less distance than a search point in the orthant algorithm.

\begin{table*}[tb!]
\centering
\begin{tabular}{|c|rrrr|rrrr|rrrr|}
\hline
& \multicolumn{4}{c|}{Domino Algorithms} & \multicolumn{4}{c|}{Orthant Algorithm} & \multicolumn{4}{c|}{Generalized CBS Algorithm} \\
$k$ & $\sigma$ & Avg & Max & Bound & $\sigma$ & Avg & Max & Bound & $\sigma$ & Avg & Max & Bound \\
\hline
1D & --- & --- & --- & --- & 0.11 & 0.47 & \textbf{0.95} & \textbf{1.00} & \textbf{0.11} & \textbf{0.47} & \textbf{0.95} & 1.05 \\
2D & 0.15 & 0.93 & 1.70 & 2.05 & \textbf{0.09} & \textbf{0.71} & \textbf{0.95} & \textbf{1.00} & 0.15 & 0.94 & 1.75 & 2.10 \\
3D & 0.18 & 1.39 & 2.35 & 3.20 & \textbf{0.07} & \textbf{0.83} & \textbf{0.95} & \textbf{1.00} & 0.18 & 1.42 & 2.45 & 3.15 \\
4D & --- & --- & --- & --- & \textbf{0.05} & \textbf{0.89} & \textbf{0.95} & \textbf{1.00} & 0.21 & 1.88 & 3.00 & 4.20 \\
5D & --- & --- & --- & --- & \textbf{0.04} & \textbf{0.92} & \textbf{0.95} & \textbf{1.00} & 0.23 & 2.33 & 3.50 & 5.25 \\
6D & --- & --- & --- & --- & \textbf{0.03} & \textbf{0.94} & \textbf{0.95} & \textbf{1.00} & 0.25 & 2.78 & 4.20 & 6.30 \\
7D & --- & --- & --- & --- & \textbf{0.02} & \textbf{0.94} & \textbf{0.95} & \textbf{1.00} & 0.27 & 3.22 & 4.65 & 7.35 \\
8D & --- & --- & --- & --- & \textbf{0.01} & \textbf{0.95} & \textbf{0.95} & \textbf{1.00} & 0.28 & 3.65 & 5.25 & 8.40 \\
\hline
\end{tabular}
\caption{Normalized number of responses ($R/\log n$) for different search algorithms across dimensions.}
\label{tab:algorithm_metrics_R}
\end{table*}

\subsection{Number of Responses (\texorpdfstring{$R$}{R})}
Up until this point, the orthant algorithm has performed poorly when compared to
the domino and generalized CBS algorithms.
Under the number of responses metric, however, the orthant algorithm is able to
shine, performing the best on all metrics and across all dimensions, as shown in
\Cref{tab:algorithm_metrics_R}, tying with the generalized CBS algorithm only in
1D.
This supports our claim that the orthant algorithm is a good choice when POI
responses either carry a high cost or pose a high risk.
It is worth noting that while the generalized CBS algorithm certainly performs
worse, especially in higher dimensions, in 1--3D it performs comparably.
Our results support the claim that the generalized CBS algorithm is a good
default choice, performing competitively with regards to the number of probes
and responses, while performing by far the best in terms of distance traveled.

\end{document}